\documentclass[11pt,letterpaper]{article}
\usepackage[letterpaper,top=1in,bottom=1in,left=1in,right=1in]{geometry}
\pdfoutput=1

\usepackage[utf8]{inputenc}
\usepackage[english]{babel}
\usepackage{amsmath,amsthm,amssymb,amsfonts,bm}
\usepackage{dsfont}
\usepackage{graphicx}
\usepackage[svgnames]{xcolor}
\usepackage{booktabs}
\usepackage{algpseudocode}
\usepackage{algorithm}
\usepackage[caption=false]{subfig}

\usepackage{hyperref}
\hypersetup{colorlinks,linkcolor=blue,urlcolor=black,citecolor=blue}
\allowdisplaybreaks

\usepackage{tikz}
\usetikzlibrary{quantikz,circuits.logic.US}
\newcommand{\multimeterD}[2]{
\gate[#1,style={and gate US,draw,inner sep=-5pt}]{#2}
}

\def\ket #1{\vert #1\rangle}
\def\bra #1{\langle #1\vert}
\newcommand{\ketbra}[2]{\ensuremath{\ket{#1}\!\bra{#2}}}
\renewcommand{\braket}[2]{\ensuremath{\langle {#1} \vert {#2} \rangle}}
\newcommand{\kett}[1]{|#1\rangle\!\rangle}
\newcommand{\bbra}[1]{\langle\!\langle#1|}
\newcommand{\kettbbra}[2]{\ensuremath{\kett{#1}\!\bbra{#2}}}
\newcommand{\bbrakett}[2]{\ensuremath{\langle\!\langle{#1}\vert{#2}\rangle\!\rangle}}
\newcommand{\bigkett}[1]{\Big|#1\Big\rangle\!\!\Big\rangle}
\newcommand{\bigbbra}[1]{\Big\langle\!\!\Big\langle#1\Big|}
\newcommand{\biggkett}[1]{\Bigg|#1\Bigg\rangle\!\!\!\Bigg\rangle}
\newcommand{\biggbbra}[1]{\Bigg\langle\!\!\!\Bigg\langle#1\Bigg|}
\newcommand{\tr}[0]{\textup{tr}}

\newtheorem{definition}{Definition}

\newtheorem{proposition}{Proposition}
\newtheorem{lemma}{Lemma}

\newtheorem{theorem}{Theorem}

\newtheorem{problem}{Problem}

\newtheorem*{theorem*}{Theorem}

\newcommand{\Hmain}[0]{\mathbb{H}_{\textup{main}}}
\newcommand{\Hanc}[0]{\mathbb{H}_{\textup{anc}}}

\newcommand{\Var}{\operatorname{Var}}

\usepackage{authblk}

\begin{document}

\title{Unitarity estimation for quantum channels}

\date{}
\author[1,2,\thanks{\href{mailto:chenka@ios.ac.cn}{chenka@ios.ac.cn}}]{Kean Chen}
\author[3,\thanks{\href{mailto:QishengWang1994@gmail.com}{QishengWang1994@gmail.com}}]{Qisheng Wang}
\author[1,2,\thanks{\href{mailto:longpx@ios.ac.cn}{longpx@ios.ac.cn}}]{Peixun Long}
\author[1,4,\thanks{\href{mailto:yingms@ios.ac.cn}{yingms@ios.ac.cn}}]{Mingsheng Ying}
\affil[1]{{\footnotesize\textit{Institute of Software, Chinese Academy of Sciences, China}}}
\affil[2]{{\footnotesize\textit{University of Chinese Academy of Sciences, China}}}
\affil[3]{{\footnotesize\textit{Graduate School of Mathematics, Nagoya University, Japan}}}
\affil[4]{{\footnotesize\textit{Department of Computer Science and Technology, Tsinghua University, China}}}

\maketitle

\begin{abstract}
Estimating the unitarity of an unknown quantum channel $\mathcal{E}$ provides information on how much it is unitary, which is a basic and important problem in quantum device certification and benchmarking.
Unitarity estimation can be performed with either coherent or incoherent access, where the former in general leads to better query complexity while the latter allows more practical implementations.
In this paper, we provide a unified framework for unitarity estimation, which induces ancilla-efficient algorithms that use $O(\epsilon^{-2})$ and $O(\sqrt{d}\cdot\epsilon^{-2})$ calls to $\mathcal{E}$ with coherent and incoherent accesses, respectively, where $d$ is the dimension of the system that $\mathcal{E}$ acts on and $\epsilon$ is the required precision. We further show that both the $d$-dependence and $\epsilon$-dependence of our algorithms are optimal.
As part of our results, we settle the query complexity of the distinguishing problem for depolarizing and unitary channels with incoherent access by giving a matching lower bound $\Omega(\sqrt{d})$, improving the prior best lower bound $\Omega(\sqrt[3]{d})$ by \hyperlink{cite.aharonov2022quantum}{Aharonov et al. (Nat. Commun. 2022)} and \hyperlink{cite.chen2022exponential}{Chen et al. (FOCS 2021)}.
\end{abstract}

\section{Introduction}
Testing and verifying properties of quantum channels is a central problem that has been extensively studied in the literature.
The standard approach for learning quantum channels is quantum process tomography~\cite{chuang1997prescription,poyatos1997complete,nielsen2002quantum}, which can reconstruct the full information of quantum channels, but is excessively costly in the dimension of the quantum system.
By contrast, we may not need the full information but certain properties of quantum channels, which is more meaningful in the learning and certification tasks.
The learning task is to partially identify some theoretical descriptions from a restricted set of possibilities that best matches the experimental results~\cite{da2011practical,aharonov2022quantum,montanaro2013survey}.
The certification task is to ensure whether an experimental quantum device acts correctly as its theoretical target~\cite{da2011practical,pfister2018verification,kliesch2021theory}. These tasks usually require much fewer resources than the full tomography, showing further significance for near-term quantum devices.

{\vskip 3pt}

\textbf{Problem statement}. In this paper, we study a basic problem --- given access to an unknown quantum channel \(\mathcal{E}\), measure how ``unitary'' \(\mathcal{E}\) is.
This problem is important both theoretically and practically. 
To see this, we consider the following scenarios that the unitarity is used to test the functionality of quantum channels:
\begin{enumerate}
\item If \(\mathcal{E}\) is a quantum gate/circuit device, one might wonder whether \(\mathcal{E}\) acts (or acts closely) as a unitary operation. This can be seen as a certification task~\cite{montanaro2013survey} without the prerequisite of knowing detailed specification of \(\mathcal{E}\), but only the knowledge that \(\mathcal{E}\) ought to be unitary.
\item If \(\mathcal{E}\) is a noisy process, one might wonder whether the dominant noise of \(\mathcal{E}\) is coherent (i.e., overrotaion or calibration errors) or incoherent (i.e., depolarizing or dephasing noise)~\cite{dirkse2019efficient}. This is useful since such two different types of noise are generally reduced in different ways~\cite{feng2016estimating,sheldon2016characterizing}. Furthermore, this also provides a tighter connection  between worst-case and average-case errors~\cite{kueng2016comparing,wallman2015bounding}; and a better bound for the gate fidelity of composite channels~\cite{carignan2019bounding}.
\end{enumerate}
The absolute unitarity was characterized in \cite{montanaro2013survey} via the well-known Choi-Jamio{\l}kowski isomorphism~\cite{choi1975completely,jamiolkowski1972linear}, in the sense that a quantum channel is unitary if and only if its Jamio{\l}kowski state is pure. Therefore, a unitarity measure can be defined as a direct application of purity:
\begin{equation} \label{11261554}
    \mathfrak{u}(\mathcal{E}) := \tr \left( \mathfrak{J}(\mathcal{E})^2 \right),
\end{equation}
where $\mathfrak{J}(\mathcal{E}) = (\mathcal{E} \otimes \mathcal{I}) (\ketbra{\Phi}{\Phi})$ is the Jamio{\l}kowski state of $\mathcal{E}$ and $\ket{\Phi}$ is the maximally entangled state. 
One can see that, $\mathfrak{u}(\mathcal{E}) \leq 1$ with equality if and only if $\mathcal{E}$ is unitary. Another \textit{closely related} definition of unitarity is proposed in \cite{wallman2015estimating} to characterize the coherence of noise. For simplicity, we use the definition of Eq.~\eqref{11261554} in the main text, and the alternative definition can be similarly handled by our methods to obtain the exactly same results (see Appendix~\ref{sec-ver}).

It is worth noting that there is already a direct approach~\cite{montanaro2013survey} for unitarity estimation, according to Eq.~\eqref{11261554}, by estimating the inner product of two copies of the Jamio{\l}kowski state \(\mathfrak{J}(\mathcal{E})\) via the SWAP test~\cite{BCWdW01}:
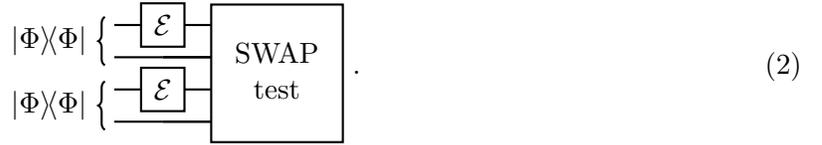
\begin{equation}
\begin{quantikz}[row sep=-0.3em, column sep=1.0em]
	 \lstick[2]{\(|\Phi\rangle\!\langle\Phi|\)}  & \gate[1]{\mathcal{E}} & \gate[wires=4,disable auto height]{\begin{array}{c}\text{SWAP}\\ \text{test}\end{array}}\\
	  \qw  &  \ghost{\mathcal{E}}\qw& \qw\\
	 \lstick[2]{\(|\Phi\rangle\!\langle\Phi|\)}  &  \gate[1]{\mathcal{E}}  & \\
	   \qw &  \ghost{\mathcal{E}}\qw & \qw
\end{quantikz}.
\end{equation}
However, such approach involves global entanglements on a large quantum system of dimension \(\Omega(d^4)\), where \(d\) is the dimension of the quantum system that \(\mathcal{E}\) acts on. This means, at least an \(\Omega(d^3)\)-dimensional ancilla system is needed, which is hardly ancilla-efficient (here the ancilla system refers to the additional quantum system other than that of the top \(\mathcal{E}\), formal definition will be given later). Furthermore, this direct approach is not practically suitable on near-term devices, since it requires so-called \textit{coherent access}~\cite{aharonov2022quantum} to $\mathcal{E}$.

{\vskip 3pt}

\textbf{Coherent/incoherent access}. To see this, we start by introducing the frameworks of learning algorithms with quantum channel access. Suppose we have an unknown quantum channel \(\mathcal{E}\) acting on a \(d\)-dimensional ``main'' system \(\Hmain\), and we also have a \(d'\)-dimensional ancilla system \(\Hanc\). To learn the channel, experiments are conducted with access to \(\mathcal{E}\), which can be divided into two categories:
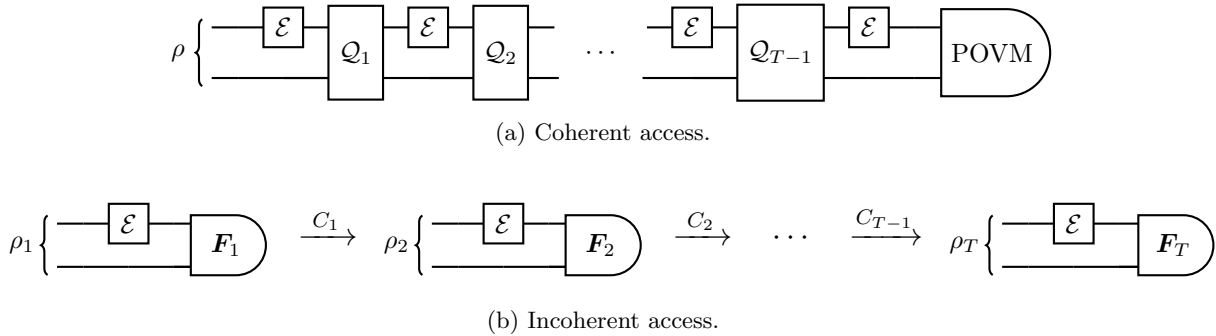
\begin{figure}[ht]
\centering
\subfloat[Coherent access.]{\label{fig:co}\small
\begin{quantikz}[row sep=0.2em, column sep=1.0em]
	 \lstick[2]{\(\rho\)}  &\qw & \gate{\mathcal{E}} & \gate[2]{\mathcal{Q}_1} & \gate{\mathcal{E}}  &\gate[2]{\mathcal{Q}_2}&\qw \midstick[2,brackets=none]{\(\cdots\)} &\gate{\mathcal{E}}& \gate[2]{\mathcal{Q}_{T-1}}& \gate{\mathcal{E}} &\qw& \gate[2,style={and gate US,draw,inner sep=-3pt}]{\textup{POVM}} \\
	   & \qw &\qw & \qw &\qw&\qw & \qw &\qw & \qw& \qw & \qw&
\end{quantikz}
}
\\
\vspace{2mm}
\subfloat[Incoherent access.]{\label{fig:inco}\small
\begin{quantikz}[row sep=-0em, column sep=1.0em]
	   \lstick[2]{\(\rho_1\)\!} \qw & \qw &\gate{\mathcal{E}} &\qw  &\multimeterD{2}{\bm{F}_1}\\
	   	  \qw & \qw & \qw& \qw &\qw
\end{quantikz}
{\large\(\quad\xrightarrow{\,\,C_1\,\,}\)}
\begin{quantikz}[row sep=-0em, column sep=1.0em]
	   \lstick[2]{\(\rho_2\)\!} \qw & \qw &\gate{\mathcal{E}} &\qw  &\multimeterD{2}{\bm{F}_2}\\
	   	  \qw & \qw & \qw& \qw &\qw
\end{quantikz}
{\large\(\quad\xrightarrow{\,\,C_2\,\,}\quad\cdots\quad\xrightarrow{C_{T-1}}\)}
\begin{quantikz}[row sep=-0em, column sep=1.0em]
	   \lstick[2]{\(\rho_T\)\!} \qw & \qw &\gate{\mathcal{E}} &\qw  &\multimeterD{2}{\bm{F}_T}\\
	   	  \qw & \qw & \qw& \qw &\qw
\end{quantikz}
}
\caption{Learning quantum channel \(\mathcal{E}\) with coherent/incoherent access.}
\label{fig-1282341}
\vspace{-3mm}
\end{figure}
\begin{enumerate}
\item \textit{Coherent access}. As shown in Fig.~\ref{fig:co}, one can prepare an arbitrary initial state \(\rho\) on \(\Hmain \otimes \Hanc\), evolve under multiple calls to \(\mathcal{E}\otimes \mathcal{I}_{\textup{anc}}\) interleaved by arbitrary quantum channels \(\mathcal{Q}_i\), and finally perform an arbitrary POVM to obtain a classical outcome. The quantum channels \(\mathcal{Q}_i\) can be seen as quantum computers processing the quantum data from previous output states.
\item \textit{Incoherent access}. As shown in Fig.~\ref{fig:inco}, one can prepare an arbitrary initial state \(\rho_i\) on \(\Hmain \otimes \Hanc\), evolve under \(\mathcal{E}\otimes \mathcal{I}_{\textup{anc}}\), and perform an arbitrary POVM \(\bm{F}_i\) to obtain a classical outcome. This experiment is repeated many times, where at each time, there is a classical computer \(C_i\) that processes the measurement data and designs the next experiment.
The algorithm is called adaptive if the selections of \(\rho_i\) and \(\bm{F}_i\) depend on previous measurement outcomes; and is called ancilla-assisted if \(d'>1\). 
\end{enumerate}

It is easy to see that the incoherent access framework is more restricted. As shown in \cite{aharonov2022quantum,chen2022exponential}, there is an exponential separation between the learning algorithms with coherent and incoherent access, in terms of the sample/query complexity.
However, the coherent access framework is practically prohibitive for near-term devices~\cite{preskill2018quantum,chen2022tight}. The reason is two-fold:
\begin{enumerate}
  \item An algorithm with coherent access often requires a larger ancilla quantum system and longer coherence time, for storing the quantum states output by the previous calls to \(\mathcal{E}\).
  \item In the sense of quantum device testing, the device \(\mathcal{E}\) under test (e.g., a potentially defective quantum chip~\cite{wang2020integrated}) can be \textit{unique}, since it might be hard or even impossible to make any copy of \(\mathcal{E}\) that has exactly the same defect. Thus, coherent access requires the ability of saving and retrieving the quantum data from the previous experiments~\cite{chen2022exponential}, thereby prohibitive without the presence of persistent quantum memory.
\end{enumerate}
Therefore, despite the potentially higher query complexity, it is still valuable to design learning algorithms for quantum channels with incoherent access.

\subsection{Our results}
In this paper, we give a unified framework for estimating the unitarity of quantum channels, with either coherent or incoherent access. Within it, we are able to prove the following upper bounds: 

\begin{theorem*}[Upper bounds, see Section~\ref{sec-12241419}] Suppose our task is to estimate the unitarity of an unknown quantum channel \(\mathcal{E}\) acting on a \(d\)-dimensional quantum system to precision \(\epsilon\).
\begin{itemize}\item \textbf{Coherent access}.
There is an algorithm for this task with coherent access using \(O(\epsilon^{-2})\) calls to \(\mathcal{E}\) and an \(O(d)\)-dimensional ancilla system. 
\item \textbf{Incoherent access}. There is an algorithm for this task with incoherent access using \(O(\sqrt{d}\cdot\epsilon^{-2})\) calls to \(\mathcal{E}\). Moreover, this algorithm is non-adaptive, non-ancilla-assisted.
\end{itemize}
\end{theorem*}
In our theorem, we only require that the success probability be higher than a constant (e.g., \(2/3\)). To amplify the success probability to \(1-\delta\), we can use the ``median trick'', i.e., repeating the algorithm for \(O(\log(1/\delta))\) times and then taking the median of all the results. This only introduces an additional factor of \(\log(1/\delta)\). Our algorithms also apply to non-trace-preserving quantum channels. Note that in the case of coherent access, the approach of applying the SWAP test directly on the Jamio{\l}kowski states~\cite{montanaro2013survey} has the same query complexity \(O(\epsilon^{-2})\) but uses a \(\Theta(d^3)\)-dimensional ancilla system. By contrast, our algorithm only uses an \(O(d)\)-dimensional ancilla system.

We can further show that the $d$-dependence and $\epsilon$-dependence in our upper bounds are optimal by giving the following lower bounds.

\begin{theorem*}[Lower bounds, see Section~\ref{sec-12241421}]
Suppose our task is to estimate the unitarity of an unknown quantum channel \(\mathcal{E}\) acting on a \(d\)-dimensional quantum system to precision \(\epsilon\).
\begin{itemize}\item \textbf{Coherent access}.
Any algorithm for this task with coherent access must use at least \(\Omega(\epsilon^{-2})\) calls to \(\mathcal{E}\).
\item \textbf{Incoherent access}.
Any algorithm for this task with incoherent access must use at least \(\Omega(\sqrt{d}+\epsilon^{-2})\) calls to \(\mathcal{E}\), even if adaptive strategies and ancilla systems are allowed.
\end{itemize}
\end{theorem*}
In the case of coherence access, our lower bound matches the previous upper bound \(O(\epsilon^{-2})\)~\cite{montanaro2013survey}; and in the case of incoherent access, compared to the upper bound \(O(\sqrt{d}\cdot\epsilon^{-2})\), our lower bound is tight for \(d\) and \(\epsilon\) separately. It is also worth noting that, as part of this result, we provide a lower bound \(\Omega(\sqrt{d})\) for distinguishing between depolarizing and unitary channels, improving the prior best lower bound \(\Omega(\sqrt[3]{d})\) obtained in \cite{aharonov2022quantum,chen2022exponential}, and matching the upper bound \(O(\sqrt{d})\) obtained in \cite{chen2022exponential}.
These results are summarized in Table~\ref{tab:main_results}.

{\vskip 5pt}

\begin{table}[ht]
\centering
\renewcommand{\arraystretch}{1.5}
\setlength{\tabcolsep}{3mm}{
\begin{tabular}{ccc}
\hline
& Coherent access & Incoherent access \\
\hline
Upper bound & \(O(\epsilon^{-2})\)\,\,\(^*\) &   \(O(\sqrt{d}\cdot\epsilon^{-2})\)  \\
Lower bound & \(\Omega(\epsilon^{-2})\) &  \(\Omega(\sqrt{d}+\epsilon^{-2})\)\,\,\(^\dag\) \\
\hline
\end{tabular}
}
\caption{Our results on the query complexity of unitarity estimation. \({}^*\) Prior upper bound \(O(\epsilon^{-2})\) by \cite{montanaro2013survey} uses a \(\Theta(d^3)\)-dimensional ancilla system; by contrast, our upper bound only uses an \(O(d)\)-dimensional ancilla system. \({}^\dag\) The prior best lower bound is \(\Omega(\sqrt[3]{d})\) by \cite{aharonov2022quantum} and \cite{chen2022exponential} (for the case of constant \(\epsilon\)).}
\label{tab:main_results}
\end{table}

{\vskip 0pt}

\textbf{Applications}. To conclude this subsection, let us briefly discuss how to apply our results on the unitarity in benchmarking quantum processes~\cite{eisert2020quantum}. Specifically, the unitarity is able to characterize how well the channel \(\mathcal{E}\) can be approximated by a unitary channel, in the closeness measure of  \textit{gate fidelity}~\cite{nielsen2002simple} (details are shown in Theorem~\ref{theorem-11252053}). This result is directly applicable in the certification task of distinguishing whether \(\mathcal{E}\) is a unitary channel or is \(\epsilon\)-far from any unitary channel~\cite{montanaro2013survey}. Furthermore, when \(\mathcal{E}\) is a noise process, the unitarity provides good estimations for the best achievable gate fidelity of \(\mathcal{E}\), in the presence of perfect unitary control~\cite{wallman2015estimating}. This information shows how well the noise can be reduced if we can ``recalibrate'' \(\mathcal{E}\) using unitary operations.

\subsection{Overview of the techniques}
Now let us give a brief overview for the main techniques used in this paper to obtain the upper/lower bounds for unitarity estimation.
\vspace{1mm}

\textbf{Upper bound}. The first of our key ideas comes from an observation that the absolute unitarity of a channel is equivalent to its \textit{purity-preservation} and \textit{orthogonality-preservation} properties. As a quantitative generalization of this observation, we show that the unitarity can be reformulated in terms of the measures of these two properties (see Theorem~\ref{thm-1217036}), which admit efficient estimations with either coherent/incoherent access. This is because such purity/orthogonality-preservation measures only involve a quantum system of much smaller dimension, compared to those methods based on Jamio{\l}kowski states~\cite{montanaro2013survey}. For coherent access, we use the SWAP test as a subroutine, obtaining the upper bound \(O(\epsilon^{-2})\), only with an \(O(d)\)-dimensional ancilla system. For incoherent access, our algorithm requires no ancilla system, and adopts an extended version of the distributed quantum inner product estimation (DQIPE)~\cite{anshu2022distributed} as a subroutine, where the extension is to handle non-trace-preserving quantum channels. The upper bound \(O(\sqrt{d}\cdot\epsilon^{-2})\) is then obtained through analyzing the errors from the DQIPE and the extra randomness of our estimators simultaneously.
\vspace{1mm}

\textbf{Lower bound}. For incoherent access, we first give an \(\Omega(\sqrt{d})\) lower bound for the \textit{depolarizing vs unitary channel problem}~\cite{aharonov2022quantum,chen2022exponential}, which can be reduced to unitarity estimation with constant precision. Our proof is based on the tree representation~\cite{chen2022exponential} of learning algorithms, but with substantial improvements on the analysis for the ensemble of Haar-random unitary matrices. First, we exploit a symmetric form for the associated probability in the tree representation via the relation between tensor and outer products. Then, we identify the \(C^*\)-algebra characterization for the eigenspaces of the Weingarten matrix~\cite{collins2022weingarten} of permutation operators, so that the structure theorem of finite-dimensional \(C^*\)-algebra~\cite{davidson1996c} can be applied, and the main technical result (see Lemma~\ref{lemma_1}) of this proof follows. To obtain the lower bound for coherent access, we consider another distinguishing problem between two quantum channels constructed by Weyl-Heisenberg operators. These channels are teleportation-covariant~\cite{pirandola2017fundamental}, thus can be simulated by teleporting the input state over their Jamio{\l}kowski states. Therefore, we can bound the difference between the outputs of any learning algorithm on the two candidate channels, by the fidelity of their Jamio{\l}kowski states. We then make use of the multiplicativity of fidelity under tensor products to obtain the \(\Omega(\epsilon^{-2})\) lower bound. Furthermore, we use this result to strengthen the lower bound for incoherent access to \(\Omega(\sqrt{d}+\epsilon^{-2})\).

\textbf{Discussion}. For incoherent access, it is possible to give a stronger lower bound by combining \(d\) with the \(\epsilon\) in a more sophisticated way. One idea is to consider the distinguishing problem for random unitary and \(\epsilon\)-depolarizing channels. However, this approach does not seem to improve the lower bound beyond \(\Omega(\sqrt{d})\), because the technique we used in this paper cannot effectively upper bound the denominator of the one-sided bound (see Eq.~\eqref{eq-318215}) for this distinguishing problem. We suspect that new techniques are needed to fully address this problem.
Another direction is to find possible improvement on the upper bound \(O(\sqrt{d}\cdot \epsilon^{-2})\). For example, can it be improved to \(O(\sqrt{d}\cdot \epsilon^{-\alpha}+\epsilon^{-2})\) for some \(\alpha<2\)? We conjecture that this is possible if a more careful and refined error analysis for Algorithm~\ref{alg-ortho} can be established.

\section{Preliminaries}
We use \(\mathbb{M}_d\) to denote the set of all \(d\times d\) complex-valued matrices. \(\mathbb{M}_d\) can be formed as a \(d^2\)-dimensional vector space \(\mathbb{C}^{d^2}\) by simply flattening the matrices to vectors (row-major). For a \(d\times d\) matrix \(A\), we will use \(\kett{A}\) to denote the corresponding element in this vector space. For example,
\begin{equation}
\kett{\ketbra{\psi}{\phi}}=\ket{\psi}\ket{\phi^*}, \quad\quad\quad \kett{ABC^\dag}=A\otimes C^* \kett{B},
\end{equation}
where \(\ket{\phi^*},C^*\) are the entry-wise complex conjugate of \(\ket{\phi},C\) (w.r.t. the computational basis), respectively. The inner product of this vector space is thus defined by \(\bbrakett{A}{B}=\tr(A^\dag B)\).

\subsection{Quantum channels}
In this paper, we consider the general case that the quantum channel \(\mathcal{E}\) can be non-trace-preserving, which covers a wide range of quantum processes such as the qubit loss~\cite{wallman2015robust}, post-selection~\cite{knill2001scheme} and quantum programs~\cite{ying2016foundations}.
\vspace{1mm}

\textbf{Matrix representation}. A quantum channel \(\mathcal{E}: \rho\mapsto \sum_i E_i\rho E_i^\dag\) acting on a \(d\)-dimensional system can be formed as a linear map \(\mathbb{M}_d\rightarrow \mathbb{M}_d\) such that \(\mathcal{E}\kett{\rho}=\kett{\mathcal{E}(\rho)}\), where its matrix form can be written as:
\begin{equation}
\mathfrak{M}(\mathcal{E})=\sum_i E_i\otimes E_i^*
\end{equation}
with $E_i^*$ denoting the entry-wise complex conjugation of $E_i$. We call \(\mathfrak{M}(\mathcal{E})\) the matrix representation of quantum channel \(\mathcal{E}\). In a slight abuse of notation, we directly use \(\mathcal{E}\) to denote \(\mathfrak{M}(\mathcal{E})\), when its meaning is clear from the context. We can see that \(\tr(\mathcal{E})=\sum_i |\tr(E_i)|^2\) and \(\tr(\mathcal{E}^\dag\mathcal{E})=\sum_{i,j}|\tr(E_i^\dag E_j)|^2\).
We will use \(\mathbb{U}_d\) to denote the set of all \(d\times d\) unitary matrices. For an unitary \(U\), we will use the calligraphic letter \(\mathcal{U}\) to denote the corresponding quantum channel \(\rho \mapsto U\rho U^\dag\).
\vspace{1mm}

\textbf{Choi-Jamio{\l}kowski isomorphism}. The Jamio{\l}kowski state of a quantum channel \(\mathcal{E}:\rho\mapsto \sum_i E_i\rho E_i^\dag\) is defined by:
\begin{equation}
\mathfrak{J}(\mathcal{E}):=(\mathcal{E}\otimes\mathcal{I})(\ketbra{\Phi}{\Phi}),
\end{equation}
where \(\ket{\Phi}=\frac{1}{\sqrt{d}}\sum_i \ket{i}\ket{i}\) is the maximally entangled state. Note that we can also write the Jamio{\l}kowski state as \(\mathfrak{J}(\mathcal{E})=\frac{1}{d}\sum_i \kettbbra{E_i}{E_i}\).

\subsection{Unitarity}
The unitarity measures how much \(\mathcal{E}\) is unitary. There is an ingenious connection between the absolute purity of states and absolute unitarity of quantum channels via Choi-Jamio{\l}kowski isomorphism, i.e., the quantum channel \(\mathcal{E}\) is a unitary channel if and only if its Jamio{\l}kowski state \(\mathfrak{J}(\mathcal{E})\) is a pure state~\cite{montanaro2013survey}. Thus the unitarity can be defined as a generalization of purity for quantum channels.
\begin{definition}[Unitarity]\label{def-127114}
Let \(\mathcal{E}\) be a quantum channel (not necessarily trace-preserving), its unitarity is defined by,
\begin{equation}
\mathfrak{u}(\mathcal{E}):=\tr\left(\mathfrak{J}(\mathcal{E})^2\right),
\end{equation}
where \(\mathfrak{J}(\mathcal{E})\) is the Jamio{\l}kowski state of \(\mathcal{E}\).
\end{definition}
The following proposition shows several important properties. It is worth noting that, the third property provides a preferred form of unitarity using the matrix representation, which will be used frequently in our paper.
\begin{proposition}
For any quantum channel \(\mathcal{E}\),
\begin{enumerate}
\item \(\mathfrak{u}(\mathcal{E})\leq 1\) with equality if and only if \(\mathcal{E}\) is unitary,
\item the unitarity is invariant under unitary transformations, i.e., \(\mathfrak{u}(\mathcal{E})=\mathfrak{u}(\mathcal{U}\mathcal{E}\mathcal{V})\) for any unitaries \(U,V\),
\item the unitarity can be defined equivalently as \(\mathfrak{u}(\mathcal{E})=\frac{1}{d^2}\tr\left(\mathcal{E}^\dag\mathcal{E}\right)\).
\end{enumerate}
\end{proposition}
\begin{proof}
The first property is obvious. The third property can be obtained by:
\begin{equation}
\begin{split}
\tr(\mathfrak{J}(\mathcal{E})^2)&=\frac{1}{d^2}\tr\left[\sum_i \kettbbra{E_i}{E_i}\sum_j \kettbbra{E_j}{E_j}\right]=\frac{1}{d^2}\sum_{i,j} \bbrakett{E_i}{E_j}\bbrakett{E_j}{E_i}\\
&=\frac{1}{d^2}\sum_{i,j}\tr\left(E_i^\dag E_j\right)\tr\left(E_j^\dag E_i\right)=\frac{1}{d^2}\sum_{i,j}\tr\left[\left(E_i\otimes E_i^*\right)^\dag \left(E_j\otimes E_j^*\right)\right]\\
&=\frac{1}{d^2}\tr\left(\mathcal{E}^\dag\mathcal{E}\right).
\end{split}
\end{equation}
Then, the second property follows immediately.
\end{proof}

The unitarity \(\mathfrak{u}(\mathcal{E})\) gives good characterizations for the maximal fidelity of \(\mathcal{E}\) to a unitary channel. Therefore, it has important applications in quantum process certification and benchmarking, which will be discussed in Section~\ref{sec-127110}. Moreover, there is an alternative definition~\cite{wallman2015estimating} of unitarity which has been used in characterizing the coherence of noise channels, and is closely related to the unitarity given in Definition~\ref{def-127114}. We will show in the Appendix~\ref{sec-ver} that, the same lower and upper bounds also apply to this alternative definition.

\section{Upper bounds}\label{sec-12241419}
\subsection{A unified and efficient framework for unitarity estimation}
We start by considering two simple but  important properties of unitary channels: purity-preservation and orthogonality-preservation. That is, any unitary channel acting on a pure state still outputs a pure state, and any unitary channel acting on a pair of orthogonal states separately still outputs a pair of orthogonal states. We will show that, the unitarity can be efficiently characterized by the quantifications of these two properties.
\begin{definition}[Purity-preservation index, \(\mathfrak{p}(\mathcal{E})\)]
The purity-preservation index \(\mathfrak{p}(\mathcal{E})\) is defined to measure how much purity the channel \(\mathcal{E}\) preserves when acting on a random pure state \(\ket{\psi}=U\ket{0}\), in which \(U\) is a Haar random unitary matrix. That is,
\begin{equation}
\mathfrak{p}(\mathcal{E})
:=\mathbb{E}_{\psi}\left[\tr\!\left(\mathcal{E}(\ketbra{\psi}{\psi})^2\right)\right].
\end{equation}
\end{definition}
\begin{definition}[Orthogonality-preservation index, \(\mathfrak{o}(\mathcal{E})\)]
The orthogonality-preservation index \(\mathfrak{o}(\mathcal{E})\) is defined to measure how much orthogonality the channel \(\mathcal{E}\) preserves when acting on a pair of random orthogonal pure states \((\ket{\psi},\ket{\phi})=(U\ket{0},U\ket{1})\), in which \(U\) is a Haar random unitary matrix. That is,
\begin{equation}
\mathfrak{o}(\mathcal{E})
:=\mathbb{E}_{\psi,\phi}\left[\tr\!\left(\mathcal{E}(\ketbra{\phi}{\phi})\mathcal{E}(\ketbra{\psi}{\psi})\right)\right].
\end{equation}
\end{definition}
Note that \(\mathfrak{p}(\mathcal{E})\leq 1\) with equality if and only if \(\mathcal{E}\) preserves the purity for all pure states. By contrast, \(\mathfrak{o}(\mathcal{E})\geq 0\) with equality if and only if \(\mathcal{E}\) preserves the orthogonality for all pairs of orthogonal pure states. Then we have an elegant connection between the unitarity $\mathfrak{u}(\mathcal{E})$ and the purity/orthogonality-preservation indices $\mathfrak{p}(\mathcal{E})$ and $\mathfrak{o}(\mathcal{E})$:

\begin{theorem}[Unitarity as purity-preservation and orthogonality-preservation]\label{thm-1217036}
\begin{equation}
\mathfrak{u}(\mathcal{E})=\mathfrak{p}(\mathcal{E})-\left(1-\frac{1}{d}\right)\mathfrak{o}(\mathcal{E}).
\end{equation}
\end{theorem}
\begin{proof}
First, we can see that
\begin{equation}
\mathfrak{p}(\mathcal{E})=\mathbb{E}_{\psi}\left[\bbra{\psi} \mathcal{E}^\dag \mathcal{E}\kett{\psi}\right]
=\mathbb{E}_{U}\left[\bbra{\rho_0} \mathcal{U}^\dag\mathcal{E}^\dag \mathcal{E}\mathcal{U}\kett{\rho_0}\right]
=\int \textup{d}U \,\, \bbra{\rho_0} \mathcal{U}^\dag \mathcal{E}^\dag \mathcal{E} \mathcal{U}\kett{\rho_0},
\end{equation}
and
\begin{equation}
\mathfrak{o}(\mathcal{E})=\mathbb{E}_{\psi,\phi}\left[\bbra{\phi} \mathcal{E}^\dag \mathcal{E}\kett{\psi}\right]
=\mathbb{E}_{U}\left[\bbra{\rho_1} \mathcal{U}^\dag \mathcal{E}^\dag \mathcal{E}\mathcal{U}\kett{\rho_0}\right]
=\int \textup{d}U \,\, \bbra{\rho_1} \mathcal{U}^\dag \mathcal{E}^\dag \mathcal{E} \mathcal{U}\kett{\rho_0},
\end{equation}
where \(\rho_0=\ketbra{0}{0}\) and \(\rho_1=\ketbra{1}{1}\).
Then, note that \(\{\mathcal{U}=U\otimes U^*\}_{U\in\mathbb{U}_d}\) is a representation on \(\mathbb{C}^{d^2}\) for the group \(\mathbb{U}_d\) of \(d\)-dimensional unitaries, which has two irreducible invariant subspaces:
\begin{equation}
\left\{\kett{cI}\,|\,c\in\mathbb{C}\right\},\quad\quad \left\{\kett{X}\,|\,\tr(X)=0, X\in \mathbb{M}_d\right\}.
\end{equation}
Define \(\overline{\mathcal{E}^\dag\mathcal{E}}:=\int \textup{d}U\,\, \mathcal{U}^\dag \mathcal{E}^\dag \mathcal{E} \mathcal{U}\).
Then \(\overline{\mathcal{E}^\dag\mathcal{E}}\) commutes with all \(\mathcal{U}\). By Schur's lemma, \(\overline{\mathcal{E}^\dag \mathcal{E}}\) is a linear combination of the projectors of these two irreducible subspaces:
\begin{equation}
\overline{\mathcal{E}^\dag\mathcal{E}}=\alpha\left(\mathcal{I}-\kettbbra{\frac{I}{\sqrt{d}}}{\frac{I}{\sqrt{d}}}\right)+ \beta\kettbbra{\frac{I}{\sqrt{d}}}{\frac{I}{\sqrt{d}}}.
\end{equation}
Thus we have
\begin{equation}
\mathfrak{p}(\mathcal{E})=\bbra{\rho_0}\,\overline{\mathcal{E}^\dag\mathcal{E}}\,\kett{\rho_0}=\left(1-\frac{1}{d}\right)\alpha+\frac{1}{d}\beta,
\end{equation}
and
\begin{equation}
\mathfrak{o}(\mathcal{E})=\bbra{\rho_1} \,\overline{\mathcal{E}^\dag\mathcal{E}}\, \kett{\rho_0}=-\frac{1}{d}\alpha+\frac{1}{d}\beta,
\end{equation}
which imply
\begin{equation}
\alpha= \mathfrak{p}(\mathcal{E})-\mathfrak{o}(\mathcal{E}),\quad\quad \beta=(d-1)\mathfrak{o}(\mathcal{E})+\mathfrak{p}(\mathcal{E}).
\end{equation}
Thus,
\begin{equation}\label{eq-1130038}
\begin{split}
\mathfrak{u}(\mathcal{E})&=\frac{1}{d^2}\tr\left(\mathcal{E}^\dag\mathcal{E}\right)=\frac{1}{d^2}\tr\left(\overline{\mathcal{E}^\dag\mathcal{E}}\right)=\left(1-\frac{1}{d^2}\right)\alpha+\frac{1}{d^2}\beta\\
&= \mathfrak{p}(\mathcal{E})-\left(1-\frac{1}{d}\right)\mathfrak{o}(\mathcal{E}).
\end{split}
\end{equation}
\end{proof}
It is worth noting that, Theorem~\ref{thm-1217036} admits efficient estimation for the unitarity \(\mathfrak{u}(\mathcal{E})\) by separately estimating both $\mathfrak{p}(\mathcal{E})$ and $\mathfrak{o}(\mathcal{E})$. This is because the purity/orthogonality-preservation indexes only involve a quantum system of smaller dimension, compared to the method~\cite{montanaro2013survey} based on Jamio{\l}kowski states.

\subsection{Coherent access}
For simplicity, we only consider the orthogonality-preservation index $\mathfrak{o}(\mathcal{E})$, and the purity-preservation index \(\mathfrak{p}(\mathcal{E})\) can be treated through a similar manner. The overall idea is simple, i.e., 1) pick a Haar random unitary \(U\) and define the input states \(\ket{\psi}=U\ket{0},\ket{\phi}=U\ket{1}\); 2) apply \(\mathcal{E}^{\otimes 2}\) on \(\ketbra{\psi}{\psi}\otimes \ketbra{\phi}{\phi}\) to obtain \(\rho\otimes \sigma\) where \(\rho=\mathcal{E}(\ketbra{\psi}{\psi})\), \(\sigma=\mathcal{E}(\ketbra{\phi}{\phi})\); and then 3) estimate \(\tr(\rho\sigma)\) using SWAP test.
It is worth noting that, the output states \(\rho,\sigma\) can be partial density operators (i.e., \(\tr(\rho),\tr(\sigma)\leq 1\)), since we do not assume the trace-preservation of \(\mathcal{E}\). However, this is not a problem since the SWAP test can be easily extended to handle this case (see Appendix~\ref{sec-01190123}).
Therefore, we can estimate \(\mathfrak{o}(\mathcal{E})\) using Algorithm~\ref{alg-ortho-co}, and its performance is guaranteed by Theorem~\ref{thm-01182227}.

\begin{algorithm}[ht]
    \caption{\raggedright Estimating \(\mathfrak{o}(\mathcal{E})\) with coherent access}\label{alg-ortho-co}
    \begin{algorithmic}[1]
    \Require number of input settings \(M\), coherent access to the quantum channel \(\mathcal{E}\) acting on a \(d\)-dimensional system
    \Ensure an estimate of $\mathfrak{o}(\mathcal{E})$
    \For {$i=1\cdots M$}
        \State sample a random unitary matrix $U_i\sim\mathbb{U}(d)$
        \State let \(\ket{\psi_i}=U_i\ket{0}\) and \(\ket{\phi_i}=U_i\ket{1}\)
        \State let \(\rho_i=\mathcal{E}(\ketbra{\psi_i}{\psi_i})\) and \(\sigma_i=\mathcal{E}(\ketbra{\phi_i}{\phi_i})\)
        \State estimate the value of $\tr\left( \rho_i \sigma_i \right)$ using the extended SWAP test on \(\rho_i\) and \(\sigma_i\), denote by \(z_i\)
    \EndFor
    \State\textbf{Return} $z:=\frac{1}{M}\sum_{i=1}^M  z_i$
    \end{algorithmic}
\end{algorithm}

\begin{theorem}[Upper bound, coherent access]\label{thm-01182227}
Set \(M=\Theta(\epsilon^{-2})\). Algorithm~\ref{alg-ortho-co} returns an \(\epsilon\)-close estimate for \(\mathfrak{o}(\mathcal{E})\) with success probability at least \(2/3\), using \(O(\epsilon^{-2})\) calls to \(\mathcal{E}\) and an \(O(d)\)-dimensional ancilla system.
\end{theorem}
\begin{proof}
In line 5 of Algorithm~\ref{alg-ortho-co}, we obtain a value \(z_i\) such that \(\mathbb{E}[z_i|U_i]=\tr(\rho_i\sigma_i)\). Thus
\begin{equation}\mathbb{E}[z_i]=\mathbb{E}_{\psi_i,\phi_i}[\tr(\mathcal{E}(\ketbra{\psi_i}{\psi_i})\mathcal{E}(\ketbra{\phi_i}{\phi_i}))]=\mathfrak{o}(\mathcal{E}).
\end{equation}
Since \(|z_i|\leq 1\), by Hoeffding's inequality, we have
\begin{equation}
\Pr \left[ \left|  z - \mathfrak{o}(\mathcal{E}) \right| < \epsilon \right] \geq 1 - 2 \exp\left( 
-M\epsilon^2/2 \right).
\end{equation}
We take \(M=\Theta(\epsilon^{-2})\), then \(1-2\exp(-M\epsilon^{2}/2)=\Omega(1)\). This means we can estimate \(\mathfrak{o}(\mathcal{E})\) to precision \(\epsilon\) using \(O(\epsilon^{-2})\) calls to \(\mathcal{E}\). We also note that, the algorithm only needs an \(O(d)\)-dimensional ancilla system to additionally store a single-copy of the output of \(\mathcal{E}\) for the SWAP test.
\end{proof}
\(\mathfrak{p}(\mathcal{E})\) can be estimated through a similar manner. Therefore, we can estimate the value of \(\mathfrak{u}(\mathcal{E})\), using \(O(\epsilon^{-2})\) calls to \(\mathcal{E}\) and an \(O(d)\)-dimensional ancilla system. By contrast, prior upper bound \(O(\epsilon^{-2})\) obtained in  \cite{montanaro2013survey} needs an \(\Omega(d^3)\)-dimensional ancilla system, since they apply SWAP test on the Jamio{\l}kowski states, which are of much larger dimension.

\subsection{Incoherent access}
For simplicity, we only consider the orthogonality-preservation index $\mathfrak{o}(\mathcal{E})$, and the purity-preservation index \(\mathfrak{p}(\mathcal{E})\) can be treated through a similar manner. The overall idea is straightforward, i.e., replacing the swap test with the distributed quantum inner product estimation (DQIPE)~\cite{anshu2022distributed}. However, since the quantum channel \(\mathcal{E}\) can be non-trace-preserving, its output states are partial density operators, i.e. \(\tr(\mathcal{E}(\rho))\leq 1\). Direct application of the original DQIPE could fail since we may require arbitrarily many experiments to obtain enough valid samples for the collision estimator. Therefore, an extended version of DQIPE (see Algorithm~\ref{alg_einnerproduct}) is adopted to accommodate the partial density operators, and we show that its output is unbiased and has the same error bounds as those of original DQIPE (the details are presented in Appendix~\ref{sec-1130109}). Then, the performance of our unitarity estimation algorithm can be guaranteed by analyzing different types of errors separately. More details are shown in Algorithm~\ref{alg-ortho} and Theorem~\ref{thm-12111608}.

\begin{algorithm}[ht]
    \caption{\raggedright Estimating \(\mathfrak{o}(\mathcal{E})\) with incoherent access}\label{alg-ortho}
    \begin{algorithmic}[1]
    \Require number of input settings \(M\), parameters $N,m$ for inner product estimation algorithm, incoherent access to the quantum channel \(\mathcal{E}\) acting on a \(d\)-dimensional system
    \Ensure an estimate of $\mathfrak{o}(\mathcal{E})$
    \For {$i=1\cdots M$}
        \State sample a random unitary matrix $U_i\sim\mathbb{U}(d)$
        \State let \(\ket{\psi_i}=U_i\ket{0}\) and \(\ket{\phi_i}=U_i\ket{1}\)
        \State let \(\rho_i=\mathcal{E}(\ketbra{\psi_i}{\psi_i})\) and \(\sigma_i=\mathcal{E}(\ketbra{\phi_i}{\phi_i})\)
        \State estimate the value of $\tr\left( \rho_i \sigma_i \right)$ by Algorithm~\ref{alg_einnerproduct} using \(2Nm\) samples of \(\rho_i\) and \(\sigma_i\), denote by \(z_i\)
    \EndFor
    \State\textbf{Return} $z:=\frac{1}{M}\sum_{i=1}^M z_i$
    \end{algorithmic}
\end{algorithm}
\begin{theorem}[Upper bound, incoherent access]\label{thm-12111608}
Set $M = \Theta(\epsilon^{-2})$, $N = \Theta(1)$ and $m = \Theta(\sqrt{d})$. Algorithm \ref{alg-ortho} returns an \(\epsilon\)-close estimate for \(\mathfrak{o}(\mathcal{E})\) with success probability at least 2/3, using \(O(\sqrt{d}\cdot\epsilon^{-2})\) calls to \(\mathcal{E}\) and no ancilla system.
\end{theorem}
\begin{proof}
In line 5 of Algorithm \ref{alg-ortho}, by Proposition \ref{tm-1128230}, we obtain a value $z_i$ such that $\mathbb{E}[z_i|U_i] = \tr\left( \rho_i \sigma_i \right)$ and
\begin{equation}
\Var\left( z_i | U_i \right) = O \left( \frac{1}{Nd} + \frac{d}{Nm^2} + \frac{1}{Nm} \right),
\end{equation}
using $2Nm$ samples of $\rho_i$ and $\sigma_i$.
Define $E_i := \mathbb{E}[ z_i|U_i]$ and $e_i := z_i - E_i$. Then the algorithm output \(z\) can be decomposed into two parts:
\begin{equation}
z = \frac 1 M \sum_{i=1}^M E_i + \frac 1 M \sum_{i=1}^M e_i = E + e,
\end{equation}
where \(E:=\frac{1}{M}\sum_{i=1}^M E_i\) and \(e:=\frac{1}{M}\sum_{i=1}^M e_i\).
We will then bound the error of \(E\) and \(e\) separately. 
\begin{enumerate}
\item Errors of $E$: note that \(E_i =\tr(\rho_i\sigma_i)\), thus $\mathbb{E}[E_i] = \mathfrak{o}(\mathcal{E})$, and $0 \leq E_i \leq 1$. By Hoeffding's inequality, we have
\begin{equation}
\Pr \left[ \left|  E - \mathfrak{o}(\mathcal{E}) \right| < \epsilon \right] \geq 1 - 2 \exp\left( 
-2M\epsilon^2 \right).
\end{equation}
\item Errors of $e$: note that $\mathbb{E} [e_i|U_i] = 0$, and $\Var\left( e_i | U_i \right) = \Var\left( z_i | U_i \right)$. By the law of total variance, we have
\begin{equation}
\begin{split}
\Var (e_i) 
& = \mathbb{E} [ \Var (e_i | U_i) ] + \Var ( \mathbb{E} [e_i|U_i] ) \\
& = \mathbb{E} [ \Var (z_i | U_i) ] + 0 \\
& \leq O \left( \frac{1}{Nd} + \frac{d}{Nm^2} + \frac{1}{Nm} \right).
\end{split}
\end{equation} Therefore, by the Chebyshev's inequality, 
\begin{equation}
\Pr \left[ \left|  e \right| < \epsilon \right] \geq 1 - \frac{\Var( e)}{\epsilon^2},
\end{equation}
where the variance of \(e\) can be bounded by
\begin{equation}
\Var ( e) = \Var \left( \frac{1}{M} \sum_{i=1}^M e_i \right) \leq O \left( \frac{1}{MNd} + \frac{d}{MNm^2} + \frac{1}{MNm} \right).
\end{equation}
\end{enumerate}
We take $M = \Theta(\epsilon^{-2})$, $N = \Theta(1)$ and $m = \Theta(\sqrt{d})$, then
\begin{equation}
\begin{split}
&\Pr \left[ \left|  E - \mathfrak{o}(\mathcal{E}) \right| < \epsilon \right] \geq 1 - 2 \exp\left(-2M\epsilon^2 \right) = \Omega(1),\\
&\Pr \left[ \left|  e \right| < \epsilon \right] \geq 1 - \Var( e)/\epsilon^2 = \Omega(1),
\end{split}
\end{equation}
which means
\begin{equation}
\Pr \left[|z-\mathfrak{o}(\mathcal{E})|\leq 2\epsilon\right]\geq \Omega(1).
\end{equation}
Therefore, the output \(z\) is an \(\epsilon\)-close estimate for \(\mathfrak{o}(\mathcal{E})\) with high probability, and the algorithm uses \(MNm = O(\sqrt{d} \cdot \epsilon^{-2})\) oracle calls.
\end{proof}
The purity-preservation index $\mathfrak{p}(\mathcal{E})$ can be estimated to precision \(\epsilon\) through a similar manner. Therefore, we can estimate the value of $\mathfrak{u}(\mathcal{E})$ using $O(\sqrt{d} \cdot \epsilon^{-2})$ oracle calls to \(\mathcal{E}\).

\section{Lower bounds}\label{sec-12241421}
\subsection{Prerequisites}
A basic tool of our hardness results is the tree representation for learning quantum channels introduced in \cite{chen2022exponential}. The tree representation targets on the learning algorithms that are provided only incoherent access/measurement to the unknown quantum channels/states, where the state of the learning algorithm is classically represented by a node. Another key point is that the tree representation well characterizes the adaptivity of the learning algorithms, since the choice of the experiment that is performed at each node can depend on all prior experimental outcomes. Therefore, the ``tree representation'' framework is typically used for the \textit{adaptive} and \textit{incoherent} setting, which is one of the settings considered in this paper. However, it is worth noting that for different learning tasks~\cite{chen2022exponential,huang2022quantum,chen2022tight,chen2022tight2}, the techniques used to analyze the tree can vary a lot.
\begin{definition}[Tree representation]
Consider a fixed quantum channel $\mathcal{E}$ acting on a $d$-dimensional subsystem of a Hilbert space $\Hmain \otimes \Hanc$ where $\Hmain$ is the ``main system'' of dimension \(d\) and $\Hanc$ is an ``ancilla system'' of dimension \(d'\). A learning algorithm with incoherent access to \(\mathcal{E}\) can be represented as a rooted tree $\mathcal{T}$ of depth $T$ such that each node encodes all measurement outcomes the algorithm has received thus far.  The tree has the following properties:
\begin{itemize}
    \item Each node $u$ has an associated probability $p^{\mathcal{E}}(u)$.
    \item The root of the tree $r$ has an associated probability $p^{\mathcal{E}}(r) = 1$.
    \item At each non-leaf node $u$, we prepare a state $|\phi_u\rangle$ on $\Hmain \otimes \Hanc$, apply the channel $\mathcal{E}$ onto \(\Hmain\), and measure a rank-1 POVM $\{w_u^v d d'\,|\psi_u^v\rangle \langle \psi_u^v|\}_v$ (which can depend on $u$) on the entire system to obtain a classical outcome $v$, where \(\sum_v w_u^v=1\) by normalization condition. Each POVM outcome $v$ corresponds to a child node $v$ of the node $u$, which is connected by the edge $e_{u, v}$. We refer to the set of child node of the node $u$ as $\mathrm{child}(u)$.
    \item If $v$ is a child node of $u$ connected by \(e_{u,v}\), then
    \begin{equation} \label{eq:prob-node}
    p^{\mathcal{E}}(v) = p^{\mathcal{E}}(u) \, w_u^v d d' \, \bra{\psi_u^v}\,(\mathcal{E} \otimes \mathcal{I}_{\text{\rm anc}})[\ketbra{\phi_u}{\phi_u}] \,\ket{\psi_u^v}\,.
    \end{equation}
    \item Each root-to-leaf path is of length $T$.  For a leaf of corresponding to node $\ell$, $p^{\mathcal{E}}(\ell)$ is the probability that the classical memory is in state $\ell$ after the learning procedure.
\end{itemize}
\end{definition}
Note that in the tree representation, pure input states and rank-\(1\) POVMs are used instead of the mixed states and general POVMs such as those shown in Fig.~\ref{fig-1282341}. However, this is without loss of generality, since any mixed state can be simulated by a pure state with additional ancilla system, and any POVM can be simulated by a rank-\(1\) POVM (see \cite{chen2022exponential} for more details). Then, we consider the following distinguishing task:
\begin{definition}[Two-hypothesis distinguishing task]
The following two events happen with equal probability:
\begin{itemize}
    \item The channel $\mathcal{E}$ is sampled from a probability distribution $D_A$ over channels.
    \item The channel $\mathcal{E}$ is sampled from a probability distribution $D_B$ over channels.
\end{itemize}
The goal is to distinguish whether $\mathcal{E}$ is sampled from $D_A$ or $D_B$.
\end{definition}
The following lemma is also needed.
\begin{lemma}[One-sided bound, see Lemma 5.4 in \cite{chen2022exponential}] \label{lem:onesided}
Consider a learning algorithm with incoherent access that is described by a rooted tree $\mathcal{T}$.
If we have
\begin{equation}\label{eq-318215}
\frac{ \mathbb{E}_{\mathcal{E}\sim D_A}[p^{\mathcal{\mathcal{E}}}(\ell)]}{\mathbb{E}_{\mathcal{E}\sim D_B}[p^{\mathcal{E}}(\ell)]} \geq 1 - \delta, \quad \forall \ell\in \mathrm{leaf}(\mathcal{T}),
\end{equation}
then the probability that the learning algorithm solves the two-hypothesis distinguishing task is upper bounded by $\delta$.
\end{lemma}

\subsection{Depolarizing vs unitary channel}
To prove the hardness of unitarity estimation, we first consider the following problem:
\begin{problem}[Depolarizing vs unitary channel]\label{depolar-unitary}
Suppose that a quantum channel \(\mathcal{E}\) acting on a \(d\)-dimensional system is one of the following with equal probability:
\begin{enumerate}
    \item \(\mathcal{E}\) is the completely depolarizing channel \(\mathcal{D}\),
    \item \(\mathcal{E}\) is the unitary channel \(\mathcal{U}:\rho\mapsto U\rho U^\dag\) for \(U\) a fixed, Haar-random unitary.
\end{enumerate}
The task is to distinguish between the above two cases.
\end{problem}
Here, the completely depolarizing channel \(\mathcal{D}\) is defined by \(\mathcal{D}(\rho):=\tr(\rho)I/d\).
Note that \(\mathfrak{u}(\mathcal{D})=\frac{1}{d^2}\) and \(\mathfrak{u}(\mathcal{U})=1\). Thus any unitarity estimation algorithm with constant precision can distinguish between the completely depolarizing channel and the random unitary channel.
This problem has been previously studied in several literatures~\cite{brandao2021models,aharonov2022quantum,chen2022exponential}. Specifically, a lower bound of \(\Omega(\sqrt[3]{d})\) is presented in \cite{aharonov2022quantum} and then strengthened by \cite{chen2022exponential} to allow for arbitrary-size ancilla system. In this paper, we further improve this lower bound to \(\Omega(\sqrt{d})\), which matches the upper bound \(O(\sqrt{d})\) by \cite{chen2022exponential}.

\begin{theorem}[Lower bound for depolarizing vs unitary channel, incoherent access]~\label{theorem-11232021}
Any learning algorithm with incoherent access requires
\begin{equation}
T\geq \Omega(\sqrt{d})
\end{equation}
oracle calls to \(\mathcal{E}: \mathbb{M}_d\rightarrow \mathbb{M}_d\) to distinguish between whether \(\mathcal{E}\) is a unitary or a maximally depolarizing channel with probability at least \(2/3\). 
\end{theorem}

\begin{proof}
First of all, let us assume that
\begin{equation}\label{eq_8}
\frac{d^T}{d(d+1)\cdots (d+T-1)}> \frac{1}{2}.
\end{equation}
This is because if not, we have
\begin{equation}
\begin{split}
\frac{1}{2}&\geq\frac{d^T}{d(d+1)\cdots (d+T-1)}=\prod_{t=0}^{T-1} \left(1+\frac{t}{d}\right)^{-1}\\
&>\prod_{t=0}^{T-1} \left(1-\frac{t}{d}\right)\geq \left(1-\frac{T}{d}\right)^{T}>1-\frac{T^2}{d},
\end{split}
\end{equation}
implying \(T\geq \Omega(\sqrt{d})\). Then we have done.

Let $\mathcal{T}$ be the tree corresponding to any given learning algorithm for the depolarizing vs unitary channel problem. By Lemma~\ref{lem:onesided} it suffices to lower bound $\mathbb{E}_U[p^{\mathcal{U}}(\ell)]/p^{\mathcal{D}}(\ell)$ for all leaves $\ell$. 
Now we consider a fixed leaf \(\ell\). Recall from Eq.~\eqref{eq:prob-node} and following the root-to-leaf path $v_0 = r, v_1,...,v_{T-1}, v_T = \ell$, the probability of the leaf $\ell$ under the channel $\mathcal{E}$ is
\begin{equation}
p^{\mathcal{E}}(\ell)=\prod_{t=0}^{T-1} \Bigl(w_{v_t}^{v_{t+1}} dd' \, \bra{\psi_{v_t}^{v_{t+1}}} \, (\mathcal{E}\otimes \mathcal{I}_{\textup{anc}})\bigl[\ketbra{\phi_{v_t}}{\phi_{v_t}}\bigr]\,\ket{\psi_{v_t}^{v_{t+1}}}\Bigr).
\end{equation}
\noindent Since the leaf \(\ell\) is fixed (and so is the path), for simplicity, we define \(\ket{\psi_t}=\ket{\psi_{v_{t}}^{v_{t+1}}}\), \(w_t=w_{v_t}^{v_{t+1}}\). Then it can be written as:
\begin{equation}
\prod_{t=0}^{T-1} \Bigl( w_{t} dd' \, \bra{\psi_t} \, (\mathcal{E}\otimes \mathcal{I}_{\textup{anc}})\bigl[\ketbra{\phi_t}{\phi_t}\bigr]\,\ket{\psi_t}\Bigr).
\end{equation}
Let \(\ket{\phi_t}=\sum_{i=0}^{d'-1} \ket{\phi_{t,i}}\ket{i}\) and \(\ket{\psi_{t}}=\sum_{i=0}^{d'-1} \ket{\psi_{t,i}}\ket{i}\) where the first and second registers correspond to the main and ancilla systems, respectively (note that \(\ket{\phi_{t,i}}\) and \(\ket{\psi_{t,i}}\) need not to be unit). Then,
\begin{equation}\label{eq_2}
\begin{split}
p^{\mathcal{E}}(\ell)&=\prod_{t=0}^{T-1} \Bigl(w_t d d' \, \sum_{i,j} \bra{\psi_{t,i}}\, \mathcal{E}\bigl[\ketbra{\phi_{t,i}}{\phi_{t,j}}\bigr] \, \ket{\psi_{t,j}}\Bigr) \\
&=\sum_{\mathbf{i},\mathbf{j}} \prod_{t=0}^{T-1}\Bigl(w_t d d' \, \bra{\psi_{t,i_t}}\, \mathcal{E}\bigl[\ketbra{\phi_{t,i_t}}{\phi_{t,j_t}}\bigr] \, \ket{\psi_{t,j_t}}\Bigr)\\
&=\sum_{\mathbf{i},\mathbf{j}} \prod_{t=0}^{T-1}\Bigl(w_t d d' \, \bra{\psi_{t,\mathbf{i}}}\, \mathcal{E}\bigl[\ketbra{\phi_{t,\mathbf{i}}}{\phi_{t,\mathbf{j}}}\bigr] \, \ket{\psi_{t,\mathbf{j}}}\Bigr),
\end{split}
\end{equation}
where \(\mathbf{i}=(i_0,\ldots,i_{T-1})\), \(\mathbf{j}=(j_0,\ldots,j_{T-1})\) go over all length-\(T\) sequences of elements in \(\{0,1,\ldots,d'-1\}\) and we define \(\ket{\psi_{t,\mathbf{i}}}=\ket{\psi_{t,i_t}}\), \(\ket{\phi_{t,\mathbf{i}}}=\ket{\phi_{t,i_t}}\). Now we consider the cases where \(\mathcal{E}\) is a Haar-random unitary or a completely depolarizing channel.
\\

\noindent\textbf{Haar-random unitary}. 
\vspace{1mm}

\noindent For fixed \(\mathbf{i}, \mathbf{j}\), we have
\begin{align}
&\mathbb{E}_U\left[\,\prod_{t=0}^{T-1} \bra{\psi_{t,\mathbf{i}}}\, \mathcal{U}\bigl[\ketbra{\phi_{t,\mathbf{i}}}{\phi_{t,\mathbf{j}}}\bigr]\, \ket{\psi_{t,\mathbf{j}}}\right]
=\mathbb{E}_U\left[\,\prod_{t=0}^{T-1} \bra{\psi_{t,\mathbf{i}}} U\ketbra{\phi_{t,\mathbf{i}}}{\phi_{t,\mathbf{j}}} U^\dag \ket{\psi_{t,\mathbf{j}}}\right]
\nonumber\\
=&\mathbb{E}_U\left[\,\prod_{t=0}^{T-1} \bra{\psi_{t,\mathbf{i}}} U\ket{\phi_{t,\mathbf{i}}}\bra{\psi^*_{t,\mathbf{j}}} U^* \ket{\phi^*_{t,\mathbf{j}}}\right]
=\mathbb{E}_U\left[\,\prod_{t=0}^{T-1} \bra{\psi_{t,\mathbf{i}}} U\ket{\phi_{t,\mathbf{i}}} \prod_{t=0}^{T-1} \bra{\psi^*_{t,\mathbf{j}}} U^*  \ket{\phi^*_{t,\mathbf{j}}}\right]
\label{eq_1}\\
=&\mathbb{E}_U\left[\,\begin{matrix}\vdots \\ \bra{\psi_{t,\mathbf{i}}}\\ \vdots\end{matrix}
\,U^{\otimes T}\,
\begin{matrix}\vdots\\ \ket{\phi_{t,\mathbf{i}}}\\ \vdots \end{matrix}\,\,
\begin{matrix} \vdots\\ \bra{\psi^*_{t,\mathbf{j}}} \\ \vdots \end{matrix} 
\,U^{*\otimes T}\,
\begin{matrix}\vdots\\\ket{\phi^*_{t,\mathbf{j}}} \\ \vdots \end{matrix}\,\right]
=\begin{matrix}\vdots \\ \bra{\psi_{t,\mathbf{i}}}\\ \vdots\\ \bra{\psi^*_{t,\mathbf{j}}} \\ \vdots \end{matrix} 
\,\mathbb{E}_U\Bigl[ U^{\otimes T}\otimes U^{*\otimes T}\Bigr]\,
\begin{matrix}\vdots\\ \ket{\phi_{t,\mathbf{i}}}\\ \vdots \\\ket{\phi^*_{t,\mathbf{j}}} \\ \vdots \end{matrix},\nonumber
\end{align}
where \(\ket{\psi^*},U^*\) denote the complex conjugates of \(\ket{\psi},U\) (w.r.t. the computational basis), and for readability, we use the vertical tensor form \(\,\begin{matrix}A \\ B \end{matrix}\,\) to denote \(A\otimes B\), and thus \(\begin{matrix}\vdots\\A_t\\\vdots \end{matrix}\) denotes \(\bigotimes_{t} A_t\). By the conclusions in representation theory~\cite{fulton2013representation}, \(\mathbb{E}_U[U^{\otimes T} \otimes U^{*\otimes T}]\) is the orthogonal projector \(\Pi_{\textup{triv}}\) onto the trivial sub-representations of the unitary representation \(\{U^{\otimes T} \otimes U^{*\otimes T}\}_U\) of the group \(\{U\}_U\). Roughly speaking, we have used a fact similar to that \(|\mathcal{G}|^{-1} \sum_{g\in\mathcal{G}} \textup{REP}(g)\) is the orthogonal projector onto the trivial sub-representations of the unitary representation REP of group \(\mathcal{G}\) (also see the proof of Proposition 1 in \cite{harrow2013church}). Due to the equivalence of the following fixed point problems,
\begin{equation}\label{eq11202347}
U^{\otimes T} \otimes U^{*\otimes T} \kett{X}=\kett{X} \quad\Longleftrightarrow\quad U^{\otimes T} X U^{\dag \otimes T} = X,
\end{equation}
we can consider the RHS of Eq.~\eqref{eq11202347}. By the Schur-Weyl duality~\cite{fulton2013representation}, we know that any solution \(X\) belongs to the linear span of \(\{P_{\pi}\, |\, \pi\in\mathbb{S}_T\}\), where \(\mathbb{S}_T\) is the symmetric group on \(T\) letters,  and \(P_\pi: \Hmain^{\otimes T}\rightarrow \Hmain^{\otimes T}\), \(\ket{x_1,\ldots,x_T}\mapsto\ket{x_{\pi^{-1}(1)},\ldots,x_{\pi^{-1}(T)}}\). Thus the support of \(\Pi_{\textup{triv}}\) is spanned by \(\kett{P_{\pi}}\). However, note that these \(\kett{P_{\pi}}\) are not orthogonal with each other. We will later give a lemma that relates the orthogonal basis of \(\Pi_{\textup{triv}}\) and \(\kett{P_{\pi}}\). But for now, let us suppose the spectral decomposition of \(\Pi_{\textup{triv}}\) is \(\sum_k \kettbbra{A_k}{A_k}\). Then, 
\begin{equation}
\begin{split}
(\ref{eq_1})
&= \begin{matrix}\vdots \\ \bra{\psi_{t,\mathbf{i}}}\\ \vdots\\ \bra{\psi^*_{t,\mathbf{j}}} \\ \vdots \end{matrix} 
\,\sum_k \kettbbra{A_k}{A_k}\,
\begin{matrix}\vdots\\ \ket{\phi_{t,\mathbf{i}}}\\ \vdots \\\ket{\phi^*_{t,\mathbf{j}}} \\ \vdots \end{matrix}
=\biggbbra{\begin{matrix}\vdots\\ \ketbra{\psi_{t,\mathbf{i}}}{\psi_{t,\mathbf{j}}} \\\vdots\end{matrix}} \,\sum_k \kettbbra{A_k}{A_k}\,\,
\biggkett{\begin{matrix}\vdots\\ \ketbra{\phi_{t,\mathbf{i}}}{\phi_{t,\mathbf{j}}} \\\vdots\end{matrix}}\\
&=\sum_k \tr\left[\begin{matrix}\vdots\\\ketbra{\psi_{t,\mathbf{j}}}{\psi_{t,\mathbf{i}}}\\\vdots\end{matrix}\,A_k\right]\,\, 
\tr\left[A_k^\dag \,\begin{matrix}\vdots\\\ketbra{\phi_{t,\mathbf{i}}}{\phi_{t,\mathbf{j}}}\\\vdots\end{matrix}\right]
=\sum_k \begin{matrix} \vdots\\\bra{\psi_{t,\mathbf{i}}}\\\vdots\end{matrix} A_k \begin{matrix}\vdots\\\ket{\psi_{t,\mathbf{j}}}\\\vdots\end{matrix}\,
\begin{matrix} \vdots\\\bra{\phi^*_{t,\mathbf{i}}}\\\vdots\end{matrix} A^*_k \begin{matrix}\vdots\\\ket{\phi^*_{t,\mathbf{j}}}\\\vdots\end{matrix}\\
&=\begin{matrix}\vdots\\\bra{\psi_{t,\mathbf{i}}}\\\vdots\\\bra{\phi^*_{t,\mathbf{i}}}\\\vdots\end{matrix}
\,\sum_k A_k\otimes A_k^*\,
\begin{matrix}\vdots\\\ket{\psi_{t,\mathbf{j}}}\\\vdots\\\ket{\phi^*_{t,\mathbf{j}}}\\\vdots\end{matrix}.
\end{split}
\end{equation}
Then, recall from Eq.~\eqref{eq_2}, we have
\begin{equation}\label{eq_3}
\mathbb{E}_U\left[p^{\mathcal{U}}(\ell)\right]
=\left(\prod_{t=0}^{T-1}w_t dd'\right) \,
\sum_{\mathbf{i}} \begin{matrix}\vdots\\\bra{\psi_{t,\mathbf{i}}}\\\vdots\\\bra{\phi^*_{t,\mathbf{i}}}\\\vdots\end{matrix}
\, \left(\sum_k A_k\otimes A_k^*\right)\,
\sum_{\mathbf{j}} \begin{matrix}\vdots\\\ket{\psi_{t,\mathbf{j}}}\\\vdots\\\ket{\phi^*_{t,\mathbf{j}}}\\\vdots\end{matrix}.
\end{equation}
Note that the LHS and RHS of \(\sum_k A_k\otimes A_k^*\) are the same vector, and we have the following lemma.
\begin{lemma}\label{lemma_1}
If \(\frac{d^T}{d(d+1)\cdots(d+T-1)}> \frac{1}{2}\), then
\begin{equation}
\sum_k A_k\otimes A_k^* \sqsupseteq \frac{1}{d(d+1)\cdots (d+T-1)}\sum_{\pi\in \mathbb{S}_T} P_{\pi}\otimes P^*_{\pi},
\end{equation}
where \(\sqsupseteq\) is the Loewner order.
\end{lemma}
The proof of the above lemma is deferred to Section~\ref{sec_1}. This lemma yields
\begin{equation}\label{eq_4}
\begin{split}
\eqref{eq_3} 
&\geq\frac{\prod_{t=0}^{T-1}w_t dd'}{d(d+1)\cdots(d+T-1)} \,\,
\sum_{\mathbf{i}} \begin{matrix}\vdots\\\bra{\psi_{t,\mathbf{i}}}\\\vdots\\\bra{\phi^*_{t,\mathbf{i}}}\\\vdots\end{matrix}
\, \left(\sum_{\pi\in\mathbb{S}_T} P_\pi\otimes P_\pi^*\right)\,
\sum_{\mathbf{j}} \begin{matrix}\vdots\\\ket{\psi_{t,\mathbf{j}}}\\\vdots\\\ket{\phi^*_{t,\mathbf{j}}}\\\vdots\end{matrix}\\
&=\frac{\prod_{t=0}^{T-1}w_t dd'}{d(d+1)\cdots(d+T-1)} \,\,
\sum_{\mathbf{i}} \begin{matrix}\vdots\\\bigbbra{\,\ketbra{\psi_{t,\mathbf{i}}}{\phi_{t,\mathbf{i}}}\,}\\\vdots\end{matrix}
\, \left(\sum_{\pi\in\mathbb{S}_T} \mathcal{P}_\pi \right)\,
\sum_{\mathbf{j}} \begin{matrix}\vdots\\\bigkett{\,\ketbra{\psi_{t,\mathbf{j}}}{\phi_{t,\mathbf{j}}}\,}\\\vdots\end{matrix},
\end{split}
\end{equation}
where \(\mathcal{P}_\pi: \mathbb{M}_d^{\otimes T}\rightarrow \mathbb{M}_d^{\otimes T}\), \(\kett{X_1}\cdots\kett{X_T}\mapsto\kett{X_{\pi^{-1}(1)}}\cdots\kett{X_{\pi^{-1}(T)}}\). \(\mathcal{P}_\pi\) can be seen as the permutation operator ``lifted'' from \(P_\pi\). Then,
\begin{equation}\label{eq_5}
\eqref{eq_4}=\frac{\prod_{t=0}^{T-1}w_t dd'}{d(d+1)\cdots(d+T-1)} \,
\begin{matrix}\vdots\\\bigbbra{\,\sum_{i=0}^{d'-1}\ketbra{\psi_{t,i}}{\phi_{t,i}}\,}\\\vdots\end{matrix}
\, \left(\sum_{\pi\in\mathbb{S}_T} \mathcal{P}_\pi \right)\,
\begin{matrix}\vdots\\\bigkett{\,\sum_{j=0}^{d'-1}\ketbra{\psi_{t,j}}{\phi_{t,j}}\,}\\\vdots\end{matrix}.
\end{equation}
Note that the LHS and RHS of \(\sum_{\pi\in\mathbb{S}_T} \mathcal{P}_\pi\) are tensor product vectors. Then, we use the following lemma about tensor product vector and permutation operators:
\begin{lemma}[Tensor product vector and permutation operators, see Lemma 5.12 in \cite{chen2022exponential}]
    For any collection of vectors $\{\ket{x_t}\}_{t=1}^T \subset \mathbb{C}^m$, \begin{equation}
        \begin{matrix}\vdots\\ \bra{x_t}\\\vdots\end{matrix} \, \sum_{\pi\in\mathbb{S}_T} P_\pi \, \begin{matrix}\vdots\\ \ket{x_t}\\\vdots\end{matrix}  \geq \begin{matrix}\vdots\\ \bra{x_t}\\\vdots\end{matrix} \, \begin{matrix}\vdots\\\ket{x_t}\\\vdots\end{matrix},
    \end{equation}
    where in this lemma, \(P_\pi\) act as permutation operators on \((\mathbb{C}^{m})^{\otimes T}\).
\end{lemma}
Thus,
\begin{equation}\label{eq_6}
\begin{split}
\eqref{eq_5}
&\geq \frac{\prod_{t=0}^{T-1}w_t dd'}{d(d+1)\cdots(d+T-1)} \,
\begin{matrix}\vdots\\\bigbbra{\,\sum_{i=0}^{d'-1}\ketbra{\psi_{t,i}}{\phi_{t,i}}\,}\\\vdots\end{matrix} \,\,
\begin{matrix}\vdots\\\bigkett{\,\sum_{j=0}^{d'-1}\ketbra{\psi_{t,j}}{\phi_{t,j}}\,}\\\vdots\end{matrix}\\
&=\frac{\prod_{t=0}^{T-1}w_t dd'}{d(d+1)\cdots(d+T-1)} \,
\sum_{\mathbf{i}} \begin{matrix}\vdots\\\bigbbra{\,\ketbra{\psi_{t,\mathbf{i}}}{\phi_{t,\mathbf{i}}}\,}\\\vdots\end{matrix} \,\,
\sum_{\mathbf{j}}\begin{matrix}\vdots\\\bigkett{\,\ketbra{\psi_{t,\mathbf{j}}}{\phi_{t,\mathbf{j}}}\,}\\\vdots\end{matrix}\\
&=\frac{\prod_{t=0}^{T-1}w_t dd'}{d(d+1)\cdots(d+T-1)} \,
\sum_{\mathbf{i},\mathbf{j}} \prod_{t=0}^{T-1} \braket{\psi_{t,\mathbf{i}}}{\psi_{t,\mathbf{j}}} \braket{\phi_{t,\mathbf{j}}}{\phi_{t,\mathbf{i}}}.
\end{split}
\end{equation}
\\

\noindent\textbf{Completely depolarizing channel}.
\vspace{1mm}

\noindent If \(\mathcal{E}\) is the completely depolarizing channel \(\mathcal{D}\), then,
\begin{equation}\label{eq_7}
\begin{split}
p^{\mathcal{D}}(\ell)
&=\sum_{\mathbf{i},\mathbf{j}} \prod_{t=0}^{T-1}\Bigl(w_t d d' \, \bra{\psi_{t,\mathbf{i}}}\, \mathcal{D}\bigl[\ketbra{\phi_{t,\mathbf{i}}}{\phi_{t,\mathbf{j}}}\bigr] \, \ket{\psi_{t,\mathbf{j}}}\Bigr)
=\sum_{\mathbf{i},\mathbf{j}}\prod_{t=0}^{T-1} \Bigl(w_t d d'\, \braket{\psi_{t,\mathbf{i}}}{\psi_{t,\mathbf{j}}} \, \braket{\phi_{t,\mathbf{j}}}{\phi_{t,\mathbf{i}}}/d \Bigr)\\
&=\left(\prod_{t=0}^{T-1} w_t d'\right)\, \sum_{\mathbf{i},\mathbf{j}}\prod_{t=0}^{T-1} \braket{\psi_{t,\mathbf{i}}}{\psi_{t,\mathbf{j}}} \, \braket{\phi_{t,\mathbf{j}}}{\phi_{t,\mathbf{i}}}.
\end{split}
\end{equation}
\\

\noindent\textbf{Final step}.
\vspace{1mm}

\noindent Combining Eq.~\eqref{eq_6} and Eq.~\eqref{eq_7}, we see that
\begin{equation}
\begin{split}
\frac{\mathbb{E}_U \bigl[p^{\mathcal{U}}(\ell)\bigr]}{p^\mathcal{D}(\ell)}&\geq \frac{d^T}{d(d+1)\cdots (d+T-1)}=\prod_{t=0}^{T-1} \left(1+\frac{t}{d}\right)^{-1}\\
&\geq \prod_{t=0}^{T-1} \left(1-\frac{t}{d}\right)\geq \left(1-\frac{T}{d}\right)^{T}\geq 1-\frac{T^2}{d}.
\end{split}
\end{equation}
Using Lemma~\ref{lem:onesided}, we have the probability that the given learning algorithm successfully distinguishes the two settings is upper bounded by
    $1 - \left(1-\frac{T^2}{d}\right)$.
    Therefore, $2/3 \leq 1 - \left(1-\frac{T^2}{d}\right)$ implying that $T \geq \Omega(\sqrt{d})$.
\end{proof}

\subsubsection{Proof of Lemma~\ref{lemma_1}}~\label{sec_1}
Define \(r_\Pi=\textup{rank}(\Pi_{\textup{triv}})\). Note that \(\Pi_{\textup{triv}}=\sum_{k\leq r_\Pi} \kettbbra{A_k}{A_k}\) is the orthogonal projector on to the subspace spanned by \(\{\kett{P_{\pi}}\}_\pi\). Let \(M\) be the matrix composed of the column vectors \(\kett{P_{\pi}}\). Then,
\begin{equation}
\textup{supp}\left(\sum_k\kettbbra{A_k}{A_k}\right)=\textup{supp}\left(\sum_\pi\kettbbra{P_{\pi}}{P_{\pi}}\right)=\textup{supp}\left(M M^\dag\right),
\end{equation}
where \(\textup{supp}(H)\) denotes the support space of the Hermitian operator \(H\).
Thus there is a singular value decomposition: \(M=UD V^\dag\) such that \(\{\kett{A_k}\}_k\) are the first \(r_\Pi\) columns of \(U\), and exactly the first \(r_\Pi\) diagonal entries of \(D\) are non-zero. Therefore,
\begin{equation}
\begin{bmatrix}\ldots &, \kett{P_{\pi_i}},& \ldots \end{bmatrix} V = \begin{bmatrix}\ldots,\kett{A_k},\ldots,0,\ldots\end{bmatrix} D,
\end{equation}
\noindent where \(\pi_i\in\mathbb{S}_{T}\) is a permutation indexed by \(i\). Thus we can express \(A_k\) using \(P_{\pi_i}\) with the matrices \(V\) and \(D\). That is,
\begin{equation}
A_k=\frac{\sum_i V_{i,k}P_{\pi_i}}{D_{k}}.
\end{equation}
Thus 
\begin{equation}\label{eq_9}
\begin{split}
\sum_{k\leq r_\Pi} A_k\otimes A_k^*
&= \sum_{k\leq r_\Pi} \frac{\sum_{i,j}V_{i,k}V^*_{j,k} P_{\pi_i}\otimes P^*_{\pi_j}}{D_k^2}
=\sum_{i,j} P_{\pi_i}\otimes P^*_{\pi_j} \sum_{k\leq r_\Pi} \frac{V_{i,k}V^*_{j,k}}{D_k^2}\\
&=\sum_{i,j} P_{\pi_i}\otimes P_{\pi_j}^* \left(V(D^\dag D)^{-1} V^\dag\right)_{i,j},
\end{split}
\end{equation}
where \((\cdot)^{-1}\) here is the pseudo inverse (however we will later show that \(D^\dag D\) is full rank). This motivates us to consider the matrix \(G=M^\dag M=V (D^\dag D) V^\dag\), where \(G_{i,j}=\bbrakett{P_{\pi_i}}{P_{\pi_j}}\) (note that the inverse of \(G\) is so-called the Weingarten matrix~\cite{collins2022weingarten}). We then prove that \(G\) is full rank, and thus is positive definite.

For any row \(i\), we have
\begin{equation}
\begin{split}
\sum_{j}G_{i,j}&=\sum_j \bbrakett{P_{\pi_i}}{P_{\pi_j}}=\sum_j \tr\left(P^\dag_{\pi_i}P_{\pi_j}\right)=\sum_j\tr\left(P_{\pi_i^{-1}}P_{\pi_j}\right)
=\tr\left(\sum_j P_{\pi_j}\right)\\
&=T! \, \dim \left(\textup{Sym}^T(\mathbb{C}^{d})\right)
=d(d+1)\cdots (d+T-1),
\end{split}
\end{equation}
where \(\textup{Sym}^T(\mathbb{C}^d)\) is the symmetric subspace~\cite{harrow2013church} of \((\mathbb{C}^d)^{\otimes T}\).
Recall from Eq.~\eqref{eq_8}, we know that 
\begin{equation}
G_{i,i}-\sum_{j\neq i}G_{i,j}=2G_{i,i}-\sum_{j}G_{i,j}=2d^T-d(d+1)\cdots (d+T-1)>0,
\end{equation}
thus \(G\) is diagonally dominant, which means \(G\) is full rank.

Then we consider the upper bound of the eigenvalues of \(G\). Suppose \(\bm{x}=(\ldots,x_i,\ldots)\) is an eigenvector of \(G\), and \(x_M=\max_i(|x_i|)\) for some index $M$, then
\begin{equation}
\lambda |x_M|= |(G\bm{x})_M|=|\sum_j G_{M,j}x_{j}|\leq \sum_j G_{M,j}|x_M|=d(d+1)\cdots (d+T-1)|x_M|.
\end{equation}
Thus for any eigenvalue \(\lambda\) of \(G\), \(\lambda\leq d(d+1)\cdots (d+T-1)\). Recall from Eq.~\eqref{eq_9}, we have
\begin{equation}\label{eq11102345}
\begin{split}
&\sum_k A_k\otimes A_k^*- \frac{1}{d(d+1)\cdots (d+T-1)} \sum_i P_{\pi_i} \otimes P^*_{\pi_i}\\
=&\sum_{i,j}P_{\pi_i}\otimes P_{\pi_j}^* \left(G^{-1}\right)_{i,j} - \frac{1}{d(d+1)\cdots (d+T-1)} \sum_i P_{\pi_i} \otimes P^*_{\pi_i}\\
=&\sum_{i,j}P_{\pi_i}\otimes P_{\pi_j}^* \left(G^{-1}-\frac{I}{d(d+1)\cdots (d+T-1)}\right)_{i,j}\\
=&\sum_{i,j}P_{\pi_i}\otimes P_{\pi_j}^* \, \left(G'\right)_{i,j},
\end{split}
\end{equation}
where \(G'=G^{-1} - \frac{I}{d(d+1)\cdots (d+T-1)}\) is positive semidefinite as the smallest eigenvalue of \(G^{-1}\) is not smaller than \(\frac{1}{d\cdots (d+T-1)}\). 
Now, we can provide some intuition here: the form of \(\sum_{i,j}P_{\pi_i}\otimes P_{\pi_j}^* \, (G')_{i,j}\) is somewhat similar to
\begin{equation}
\begin{bmatrix}\ldots,\kett{P_{\pi_i}},\ldots\end{bmatrix} G' \begin{bmatrix} \ldots, \kett{P_{\pi_j}}, \ldots \end{bmatrix}^{\dag},
\end{equation}
except that \(P_{\pi_i},P_{\pi_j}\) are associated through tensor product (i.e., \(P_{\pi_i}\otimes P_{\pi_j}^*\)) but not outer product (i.e., \(\kettbbra{P_{\pi_i}}{P_{\pi_j}}\)). Thus we cannot directly conclude that Eq.~\eqref{eq11102345} is positive semidefinite.

We will solve this difficulty by characterizing the eigenspaces of \(G\) (thus is also the eigenspaces of \(G'\)). 
First, let us consider the matrix space \(\mathbb{P}:=\textup{span}(\{P_{\pi_i}\}_i)\). Since we have proved that \(G\) is full rank, these elements \(P_{\pi_i}\) are linear independent. Thus there is an bijection \(\gamma: \mathbb{C}^{T!}\rightarrow \mathbb{P}\), where 
\begin{equation}
\gamma\colon \bm{x}=\left[x_1,\ldots,x_{T!}\right]^{\textup{T}} \mapsto \sum_i x_i P_{\pi_i}.
\end{equation}
Then, for an eigenvalue \(\lambda\) of \(G\), the corresponding eigenspace \(E_\lambda\) has the following properties:
\begin{lemma}\label{lemma_221111109}
Suppose \(X,Y \in \mathbb{P}\).
\begin{enumerate}
    \item If \(X,Y\in \gamma(E_\lambda)\), then \(aX+bY \in \gamma(E_\lambda)\), where \(a,b\) are complex numbers.
    \item If \(X,Y\in \gamma(E_\lambda)\), then \(XY\in \gamma(E_\lambda)\).
    \item If \(X\in \gamma(E_\lambda)\), then \(X^\dag\in \gamma(E_\lambda)\).
    \item If \(X,Y\in \gamma(E_\lambda)\), then \(\bbrakett{X}{Y}:=\tr(X^\dag Y)=\lambda \,  \braket{\gamma^{-1}(X)}{\gamma^{-1}(Y)}\).
\end{enumerate}
\end{lemma}
\begin{proof}
\quad

\begin{enumerate}
\item The first property is obvious by the linearity of \(\gamma\).

\item Suppose \(\gamma^{-1}(X)=\left[\ldots,x_i,\ldots\right]^{\textup{T}}\) and \(\gamma^{-1}(Y)=\left[\ldots,y_i,\ldots\right]^{\textup{T}}\). Then we have 
\begin{equation}
\lambda x_i = \sum_j G_{i,j} x_j= \tr\left(P_{\pi_i}^\dag \sum_j P_{\pi_j}\right)x_j,\quad\quad 
\lambda y_i = \sum_j G_{i,j} y_j= \tr\left(P_{\pi_i}^\dag \sum_j P_{\pi_j}\right)y_j,
\end{equation}
and
\begin{equation}
XY=\sum_{i,j} x_iy_j P_{\pi_i} P_{\pi_j}=\sum_k P_{\pi_k} \sum_{j} x_{k j^{-1}}y_{j},
\end{equation}
where in a slight abuse of notation, we use \(kj^{-1}\) to denotes the index of the permutation \(\pi_k \pi_j^{-1}\). Thus we have
\begin{equation}
\begin{split}
\bigl(G\,\gamma^{-1}(XY)\bigr)_i
&=\sum_k G_{i,k} \sum_j x_{kj^{-1}}y_j
=\sum_j y_j \sum_k G_{i,k}x_{kj^{-1}}\\
&=\sum_j y_j \sum_k \tr\left(P_{\pi_i}^\dag P_{\pi_k}\right)x_{kj^{-1}}
=\sum_j y_j \sum_k \tr\left(P^\dag_{\pi_i} P_{\pi_{kj^{-1}}} P_{\pi_j}\right) x_{kj^{-1}}\\
&=\sum_j y_j \sum_k \tr\left(\left(P_{\pi_i} P_{\pi_j}^\dag\right)^\dag P_{\pi_{kj^{-1}}}\right) x_{kj^{-1}}
\\
&=\sum_j y_j \sum_k \tr\left(P^\dag_{\pi_{ij^{-1}}}P_{\pi_{kj^{-1}}}\right)x_{kj^{-1}}\\
&=\lambda \sum_j y_j x_{ij^{-1}}=\lambda \left(\gamma^{-1}(XY)\right)_i.
\end{split}
\end{equation}
Therefore, \(G\,\gamma^{-1}(XY)=\lambda \gamma^{-1}(XY)\) and thus \(XY\in \gamma(E_\lambda)\).

\item Similarly, let \(\gamma^{-1}(X)=\left[\ldots,x_i,\ldots\right]^{\textup{T}}\), then
\begin{equation}
X^\dag= \sum_i x_i^* P_{\pi_i}^\dag=\sum_i x_i^* P_{\pi_{i^{-1}}}=\sum_j x_{j^{-1}}^*  P_{\pi_j},
\end{equation}
Then,
\begin{equation}
\begin{split}
\left(G\, \gamma^{-1}\left(X^\dag\right)\right)_i
&= \sum_j G_{i,j} x_{j^{-1}}^*=\sum_j \tr\left(P_{\pi_i}^\dag P_{\pi_j}\right)x_{j^{-1}}^*\\
&= \sum_j \tr\left(P_{\pi_i}^\dag P_{\pi_{j^{-1}}}\right) x_{j}^*=\sum_j \tr\left(P_{\pi_i}^\dag P_{\pi_j}^\dag\right) x_j^*\\
&=\left(\sum_j \tr\left(P_{\pi_i}P_{\pi_j}\right) x_j\right)^*=\left(\sum_j \tr\left(P^\dag_{\pi_{i^{-1}}}P_{\pi_j}\right) x_j\right)^*\\
&=\left(\lambda x_{i^{-1}}\right)^*=\lambda x_{i^{-1}}^*=\lambda \left(\gamma^{-1}\left(X^\dag\right)\right)_i.
\end{split}
\end{equation}
Therefore, \(G \, \gamma^{-1}(X^\dag)=\lambda \gamma^{-1}(X^\dag)\) and thus \(X^\dag \in \gamma(E_\lambda)\).

\item Let \(\gamma^{-1}(X)=\left[\ldots,x_i,\ldots\right]^{\textup{T}}\) and \(\gamma^{-1}(Y)=\left[\ldots,y_i,\ldots\right]^{\textup{T}}\). Then we have
\begin{equation}
\begin{split}
\tr\left(X^\dag Y\right)
&=\sum_{i,j} x_i^* y_j \tr\left(P_{\pi_i}^\dag P_{\pi_j}\right)=\sum_i x_i^* \sum_j \tr\left(P_{\pi_i}^\dag P_{\pi_j}\right) y_j\\
&=\lambda \sum_i x_i^* y_i=\lambda\,  \braket{\gamma^{-1}(X)}{\gamma^{-1}(Y)}.
\end{split}
\end{equation}
\end{enumerate}
\end{proof}
\noindent The first three properties ensure that \(\gamma(E_\lambda)\) forms a \(C^*\)-algebra (equipped with the matrix operator norm), and the fourth property ensures that the inner products in \(E_\lambda\) and \(\gamma(E_\lambda)\) coincide (up to a constant scalar \(\lambda\)).

Now let us go back to Eq.~\eqref{eq11102345}. Suppose \(\sigma(G)\) is the set of the eigenvalues of \(G\), and \(\perp\!\!(E_\lambda)\) is an orthonormal basis of \(E_\lambda\). We define \(f(\lambda)=\lambda^{-1}-\frac{1}{d(d+1)\cdots(d+T-1)}\), then for any \(\lambda\in\sigma(G)\), \(f(\lambda)\) is the corresponding non-negative eigenvalue of \(G'\). Thus,
\begin{equation}\label{1111141}
\begin{split}
\eqref{eq11102345}
&=\sum_{i,j}P_{\pi_i}\otimes P_{\pi_j}^* \left(\sum_{\lambda\in\sigma(G)} f(\lambda) \sum_{\bm{x}\in\perp(E_\lambda)} \bm{x}^\dag \bm{x}\right)_{ij}\\
&= \sum_{\lambda\in \sigma(G)} f(\lambda) \left[\sum_{\bm{x}\in \perp(E_\lambda)} \left(\sum_{i} x_i P_{\pi_i} \right)\otimes \left(\sum_{j}x^*_j P^*_{\pi_j}\right)\right]\\
&=\sum_{\lambda\in\sigma(G)} f(\lambda) \left[\sum_{\bm{x}\in \perp(E_\lambda)} \gamma(\bm{x}) \otimes \gamma(\bm{x})^*\right]\\
&=\sum_{\lambda\in\sigma(G)} f(\lambda) \left[\sum_{X\in \perp(\gamma(E_\lambda))} X\otimes X^*\right],
\end{split}
\end{equation}
where the third equality is due to the property 4 of Lemma~\ref{lemma_221111109}.

Then, the finite-dimensional \(C^*\)-algebra characterization of \(\gamma(E_\lambda)\) suffices to prove the positive semidefiniteness of \(\sum_{X\in \perp(\gamma(E_\lambda))} X\otimes X^*\). To show the intuition, let us consider a special case: if \(\gamma(E_\lambda)\) is 1-dimensional, then for a non-zero \(X\in\gamma(E_\lambda)\), we have \(X^\dag X\in \gamma(E_\lambda),\) and thus \(X^\dag X = cX\) where \(0\neq c\in\mathbb{C}\); therefore \(X\otimes X^*= \frac{1}{|c|^2} (X^\dag X) \otimes (X^\dag X)^*\sqsupseteq 0\). Now let us consider the general case that \(\dim(\gamma(E_\lambda))>1\). Note that \(\gamma(E_\lambda)\) (as a collection of matrices) is the identity representation of itself (as a \(C^*\)-algebra). Thus, the representation \(\gamma(E_\lambda)\) is unitarily equivalent to a direct sum of non-zero irreducible representations and zero representations~\cite{davidson1996c}, where those isomorphic irreducible representations are grouped. That is, there exists a unitary matrix \(U\) such that for any \(X\in \gamma(E_\lambda)\),
\begin{equation}\label{1111212}
X=U\left[\left(\bigoplus_i \rho_i(X)\otimes I_i\right)\oplus 0\right] U^\dag,
\end{equation}
where \(\rho_i\) is the \(i\)-th non-zero irreducible representation and \(0\) is a zero matrix. Then, by the density theorem (see, for example, \cite{etingof2011introduction}), each \(\rho_i\) forms a full matrix algebra. Thus, by the freedom of the choice of \(\perp\!(\gamma(E_\lambda))\), we can choose a simple orthonormal basis \(\{X_{i,j,k}\}_{i,j,k}\) such that,
\begin{equation}
X_{i,j,k}=\frac{1}{\sqrt{\tr(I_i)}}U \begin{bmatrix}
\ddots &   &  & & \\
  & 0 &  & &\\
  &   &  \ketbra{j}{k}\otimes I_i & &\\
  & & & 0&\\
  & &&&\ddots
\end{bmatrix}U^\dag.
\end{equation}
That is, for each basis element \(X_{i,j,k}\), its representation is only non-zero in exactly one summand in Eq.~\eqref{1111212} (say, the \(i\)-th summand), and in this summand, its representation is \(\ketbra{j}{k}\otimes I_i/\sqrt{\tr(I_i)}\). Thus
\begin{equation}\label{1111224}
\begin{split}
&\sum_{X\in \perp(\gamma(E_\lambda))} X\otimes X^*=\sum_{i,j,k} X_{i,j,k}\otimes X_{i,j,k}^*\\
=&\sum_{i}\frac{1}{\tr(I_i)}(U\otimes U^*)
\begin{bmatrix}
\ddots &   &  & & \\
  & 0 &  & &\\
  & &\sum_{j,k} \left(\ketbra{j}{k}\otimes I_i\right) \otimes \left(\ketbra{j}{k}\otimes I_i\right) & &\\
  & & & 0&\\
  & &&&\ddots
\end{bmatrix}
(U \otimes U^*)^\dag.
\end{split}
\end{equation}
It is easy to verify that 
\begin{equation}
\sum_{j,k} (\ketbra{j}{k}\otimes I_i)^{\otimes 2}=\left(\sum_j \ket{j}\otimes I_i\otimes \ket{j}\otimes I_i\right)\left(\sum_k \bra{k}\otimes I_i\otimes \bra{k}\otimes I_i\right)\sqsupseteq 0.
\end{equation}
Thus \(\sum_{X\in \perp(\gamma(E_\lambda))} X\otimes X^*\) is positive semidefinite. Since we have \(f(\lambda)\geq 0\) for \(\lambda \in \sigma(G)\), we can conclude that
\begin{equation}
\sum_k A_k\otimes A_k^*- \frac{1}{d(d+1)\cdots (d+T-1)} \sum_i P_{\pi_i} \otimes P^*_{\pi_i}=\eqref{1111141}\sqsupseteq 0.
\end{equation}

\subsection{Unitarity estimation}
Then, to further lower bound the unitarity estimation, we consider the following problem:
\begin{problem}\label{pro-11232006}
Suppose that a quantum channel \(\mathcal{E}\) is one of the following with equal probability:
\begin{enumerate}
    \item \(\mathcal{E}=\mathcal{E}_1\) where \(\mathcal{E}_1:=a\mathcal{I}+b\mathcal{X}\),
    \item \(\mathcal{E}=\mathcal{E}_2\) where \(\mathcal{E}_2:=(a+\epsilon)\mathcal{I}+(b-\epsilon)\mathcal{X}\),
\end{enumerate}
where \(a=2/3,b=1/3\), \(\mathcal{X}(\rho)=X\rho X^\dag\) and \(X:\ket{i}\mapsto \ket{i+1 \,(\textup{mod}\,\, d)}\) for \(i=0,\ldots,d-1\). The task is to distinguish between the above two cases.
\end{problem}
We have the following theorem:
\begin{theorem}[Lower bound for Problem~\ref{pro-11232006}, coherent access]\label{theorem-11232020}
Any learning algorithm, even with coherent access, requires
\begin{equation}
T\geq \Omega(\epsilon^{-2})
\end{equation}
oracle calls to the quantum channel \(\mathcal{E}: \mathbb{M}_d\rightarrow \mathbb{M}_d\) to solve problem \ref{pro-11232006} with probability at least \(2/3\). 
\end{theorem}
\begin{proof}
Any learning algorithm with coherent access for an unknown quantum channel \(\mathcal{E}\) has the following form:
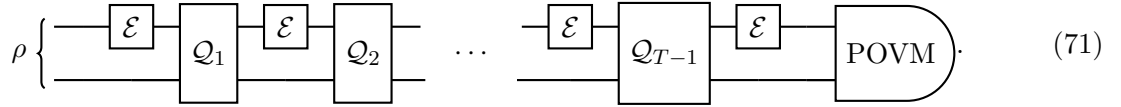
\begin{equation}\label{eq-124513}
\begin{quantikz}[row sep=0.2em, column sep=1.0em]
	 \lstick[2]{\(\rho\)}  &\qw & \gate{\mathcal{E}} & \gate[2]{\mathcal{Q}_1} & \gate{\mathcal{E}}  &\gate[2]{\mathcal{Q}_2}&\qw \midstick[2,brackets=none]{\(\cdots\)} &\gate{\mathcal{E}}& \gate[2]{\mathcal{Q}_{T-1}}& \gate{\mathcal{E}} &\qw&\gate[2,style={and gate US,draw,inner sep=-3pt}]{\textup{POVM}} \\
	   & \qw &\qw & \qw &\qw&\qw & \qw &\qw & \qw& \qw&\qw&\qw
\end{quantikz}.
\end{equation}
This algorithm acts on the quantum system \(\Hmain\otimes\Hanc\), where $\Hmain$ is the main system (the first line), $\Hanc$ is an ancilla system (the second line), and the dimension of the ancilla system is not limited. The algorithm is defined by the initial state \(\rho\) and transformations \(\mathcal{E}\otimes \mathcal{I},\mathcal{Q}_1,\mathcal{E}\otimes \mathcal{I},\mathcal{Q}_2,\ldots,\mathcal{E}\otimes\mathcal{I},\mathcal{Q}_{T-1},\mathcal{E}\otimes \mathcal{I}\). The transformations \(\mathcal{Q}_i\) are quantum channels. The initial state \(\rho\), transformation \(\mathcal{Q}_i\) and the POVM are independent of \(\mathcal{E}\). The algorithm consists of performing these transformations on the initial state \(\rho\) and then measuring the result. 

Note that \(I,X\) belong to the Weyl-Heisenberg group (or Pauli group if \(d=2\)), so that the channels \(\mathcal{E}_1\) and \(\mathcal{E}_2\) are both teleportation-covariant~\cite{pirandola2017fundamental}, which means they can be simulated by teleporting the input state over their Jamio{\l}kowski states. That is, there is a ``universal'' simulator \(\mathcal{S}\) based on quantum teleportation, such that for any state \(\sigma\) and any quantum channel \(\mathcal{E}=\sum_i E_i\circ E_i^\dag\) in which each \(E_i\) is proportional to a Weyl-Heisenberg operator, we have
\begin{equation}
\begin{quantikz}[row sep=0.5em, column sep=1.0em]
	 \lstick[2]{\(\mathfrak{J}(\mathcal{E})\)} \qw	&\qw &\gate[3]{\mathcal{S}} &\qw& \qw\rstick[wires=2]{trace out}\\
	   \qw &\qw &  & \qw & \qw\\
	   \lstick[2]{\(\sigma\)} \qw & \qw &\qw &\qw  & \qw\\
	   	  \qw & \qw & \qw& \qw & \qw\\
\end{quantikz}
\quad=\quad
\begin{quantikz}[row sep=0.8em, column sep=1.0em]
	   \lstick[2]{\(\sigma\)} \qw & \qw &\gate{\mathcal{E}} &\qw  & \qw\\
	   	  \qw & \qw & \qw& \qw & \qw\\
\end{quantikz}.
\end{equation}
Therefore, Eq.~\eqref{eq-124513} can be written as
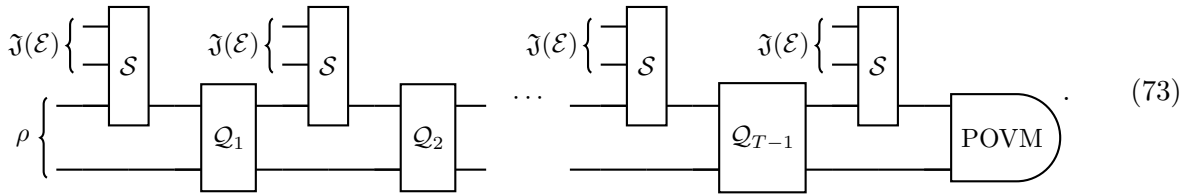
\begin{equation}\label{eq-124545}{\small
\begin{quantikz}[row sep=0.7em, column sep=1.0em]
     & \lstick[2]{\(\mathfrak{J}(\mathcal{E})\)\!} & \gate[3]{\mathcal{S}} &&&\lstick[2]{\(\mathfrak{J}(\mathcal{E})\)\!} & \gate[3]{\mathcal{S}}&&&\midstick[4,brackets=none]{\(\cdots\)}& \lstick[2]{\(\mathfrak{J}(\mathcal{E})\)\!} & \gate[3]{\mathcal{S}} &&& \lstick[2]{\(\mathfrak{J}(\mathcal{E})\)\!} & \gate[3]{\mathcal{S}}\\
     &&&&&&&&&&&&&&&\\
	 \lstick[2]{\(\rho\)}&\qw & \qw &\qw& \gate[2]{\mathcal{Q}_1} &\qw &\qw& \qw&\gate[2]{\mathcal{Q}_2} &\qw & \qw&\qw& \qw&\gate[2]{\mathcal{Q}_{T-1}}& \qw & \qw & \qw&\gate[2,style={and gate US,draw,inner sep=-3pt}]{\textup{POVM}} \\
	 	&\qw &\qw&\qw&\qw & \qw &\qw & \qw & \qw &\qw&\qw &\qw &\qw&\qw&\qw&\qw&\qw&\qw
\end{quantikz}}.
\end{equation}
Let \(\mathcal{F}\) denote the overall quantum channel (excluding the POVM) in Eq.~\eqref{eq-124545}. Then, the output of this channel is of the form \(\mathcal{F}(\mathfrak{J}(\mathcal{E})^{\otimes T}\otimes \rho)\). Thus, the trace distance between the output of the learning algorithm on different channels \(\mathcal{E}_1\), \(\mathcal{E}_2\) can be bounded by,
\begin{equation}\label{eq-124600}
\begin{split}
\left\|\mathcal{F}\left(\mathfrak{J}(\mathcal{E}_1)^{\otimes T}\otimes \rho\right)-\mathcal{F}\left(\mathfrak{J}(\mathcal{E}_2)^{\otimes T}\otimes \rho\right)\right\|_1
&\leq\left\|\mathfrak{J}(\mathcal{E}_1)^{\otimes T}\otimes \rho -\mathfrak{J}(\mathcal{E}_2)^{\otimes T}\otimes\rho\right\|_1\\
&=\left\|\mathfrak{J}(\mathcal{E}_1)^{\otimes T}-\mathfrak{J}(\mathcal{E}_2)^{\otimes T}\right\|_1,
\end{split}
\end{equation}
where the inequality is due to the contractivity of trace norm under quantum operations.
This motivates us to consider the Jamio{\l}kowski states of \(\mathcal{E}_1,\mathcal{E}_2\):
\begin{equation}
\mathfrak{J}(\mathcal{E}_1)=a\kettbbra{\frac{I}{\sqrt{d}}}{\frac{I}{\sqrt{d}}}+b\kettbbra{\frac{X}{\sqrt{d}}}{\frac{X}{\sqrt{d}}},\quad\quad
\mathfrak{J}(\mathcal{E}_2)=(a+\epsilon)\kettbbra{\frac{I}{\sqrt{d}}}{\frac{I}{\sqrt{d}}}+(b-\epsilon)\kettbbra{\frac{X}{\sqrt{d}}}{\frac{X}{\sqrt{d}}}.
\end{equation}
Thus the fidelity of \(\mathfrak{J}(\mathcal{E}_1)\) and \(\mathfrak{J}(\mathcal{E}_2)\) is
\begin{equation}
\begin{split}
F\left(\mathfrak{J}(\mathcal{E}_1),\mathfrak{J}(\mathcal{E}_2)\right)&=\sqrt{a(a+\epsilon)}+\sqrt{b(b-\epsilon)}=1+\left(\sqrt{a(a+\epsilon)}-a\right)+\left(\sqrt{b(b-\epsilon)}-b\right)\\
&=1+\frac{a\epsilon}{\sqrt{a(a+\epsilon)}+a}-\frac{b\epsilon}{\sqrt{b(b-\epsilon)}+b}=1-\epsilon\left( \frac{1}{\sqrt{1-\epsilon/b}+1}-\frac{1}{\sqrt{1+\epsilon/a}+1}\right)\\
&=1-\epsilon \left(\sqrt{1+\epsilon/a}-\sqrt{1-\epsilon/b}\right) O(1)=1-\epsilon\left(\frac{\epsilon/a+\epsilon/b}{\sqrt{1+\epsilon/a}+\sqrt{1-\epsilon/b}}\right)O(1)\\
&=1-\epsilon^2 O(1)=1-O(\epsilon^2).
\end{split}
\end{equation}
Due to the Holevo–Helstrom theorem~\cite{helstrom1969quantum} (also see, for example, 
Theorem 3.4 in \cite{watrous2018theory}), if a learning algorithm can distinguish \(\mathcal{E}_1,\mathcal{E}_2\), then the trace distance between the output of the learning algorithm on \(\mathcal{E}_1\) and \(\mathcal{E}_2\) must be larger than a constant. Thus by Eq.~\eqref{eq-124600}, we have
\begin{equation}
\Omega(1)\leq\|\mathfrak{J}(\mathcal{E}_1)^{\otimes T}-\mathfrak{J}(\mathcal{E}_2)^{\otimes T}\|_1\leq \sqrt{1-F\left(\mathfrak{J}(\mathcal{E}_1)^{\otimes T},\mathfrak{J}(\mathcal{E}_2)^{\otimes T}\right)^2}.
\end{equation}
Therefore,
\begin{equation}
\begin{split}
1-\Omega(1)&\geq F\left(\mathfrak{J}(\mathcal{E}_1)^{\otimes T},\mathfrak{J} (\mathcal{E}_2)^{\otimes T}\right)\\
&=F\left(\mathfrak{J}(\mathcal{E}_1),\mathfrak{J}(\mathcal{E}_2)\right)^T=\left(1-O(\epsilon^2)\right)^T\\
&\geq 1-T O(\epsilon^2).
\end{split}
\end{equation}
Thus
\begin{equation}
T\geq \Omega(\epsilon^{-2}).
\end{equation}
\end{proof}
Then, we have the lower bounds for unitarity estimation.
\begin{theorem}[Lower bound for unitarity estimation, coherent access]\label{thm-12111722}
Any learning algorithm with coherent access requires
\begin{equation}
T\geq \Omega(\epsilon^{-2})
\end{equation}
oracle calls to \(\mathcal{E}: \mathbb{M}_d\rightarrow \mathbb{M}_d\) to estimate the unitarity of \(\mathcal{E}\) to precision \(\epsilon\) with probability at least \(2/3\). 
\end{theorem}
\begin{proof}
Note that any unitarity estimation algorithm to precision \(O(\epsilon)\) can solve Problem~\ref{pro-11232006}.
This is because:
\begin{equation}
\begin{split}
\mathfrak{u}(\mathcal{E}_2)-\mathfrak{u}(\mathcal{E}_1)&=\tr(\mathcal{E}_2^\dag \mathcal{E}_2)/d^2-\tr(\mathcal{E}_1^\dag \mathcal{E}_1)/d^2=(a+\epsilon)^2+(b-\epsilon)^2-a^2-b^2\\
&=(2a-2b)\epsilon+2\epsilon^2=\Omega(\epsilon).
\end{split}
\end{equation}
Thus the lower bound for problem~\ref{pro-11232006} applies to unitarity estimation.
\end{proof}
\begin{theorem}[Lower bound for unitarity estimation, incoherent access]
Any learning algorithm with incoherent access requires
\begin{equation}
T\geq \Omega(\sqrt{d}+\epsilon^{-2})
\end{equation}
oracle calls to \(\mathcal{E}: \mathbb{M}_d\rightarrow \mathbb{M}_d\) to estimate the unitarity of \(\mathcal{E}\) to precision \(\epsilon\) with probability at least \(2/3\). 
\end{theorem}
\begin{proof}
We can combine the previous lower bound for incoherent access in Theorem~\ref{theorem-11232021} with the lower bound in Theorem~\ref{thm-12111722} to obtain the lower bound \(\Omega(\sqrt{d}+\epsilon^{-2})\) for unitarity estimation with incoherent access.
\end{proof}

\section{Benchmarking quantum processes with unitarity}\label{sec-127110}
Unitarity has shown various applications in benchmarking quantum processes. For example, it can be used in a specific certification task~\cite{montanaro2013survey}, which is to distinguish whether an unknown channel is a unitary channel or is \(\epsilon\)-far from any unitary channel; it also provides a lower bound for the best achievable gate infidelity for a noise channel, assuming perfect unitary control~\cite{wallman2015estimating}. Here, we introduce a holistic view of utilizing unitarity in benchmarking quantum process.

Consider the following problem: given a quantum channel \(\mathcal{E}\), how well it can be approximated by a unitary channel \(\mathcal{U}\). In this problem, a quantification of the closeness between quantum channels is needed. In the framework of quantum certification and benchmarking~\cite{emerson2005scalable,magesan2011scalable,eisert2020quantum}, the (average) gate fidelity \(F_\textup{a}\)~\cite{nielsen2002simple,emerson2005scalable} is widely used as a closeness measure of quantum channels, which is defined as
\begin{equation}
F_\textup{a}(\mathcal{U},\mathcal{E})=\int \textup{d}\psi\,\, \bra{\psi}U^\dag \mathcal{E}\left(\ketbra{\psi}{\psi}\right) U\ket{\psi}.
\end{equation}
Then, adopting this measure \(F_\textup{a}\), the \textit{unitary approximability} of \(\mathcal{E}\) can be written as:
\begin{equation}
\sup_{U\in\mathbb{U}_d} F_{\textup{a}}(\mathcal{U},\mathcal{E}).
\end{equation}
We show that the unitarity \(\mathfrak{u}(\mathcal{E})\) provides a good estimate for the unitary approximability of \(\mathcal{E}\), and the proof is given in Appendix~\ref{sec-01191541}.
\begin{theorem}\label{theorem-11252053}
Suppose \(\mathcal{E}\) is a quantum channel acting on a \(d\)-dimensional system. For convenience, let \(\mathfrak{u}\) stands for \(\mathfrak{u}(\mathcal{E})\). Then,
\begin{equation}\label{eq-12131554}
\frac{d}{d+1}\mathfrak{u}^2+\frac{1}{d+1}\sqrt{\mathfrak{u}}\leq \sup_{U\in \mathbb{U}_d} F_\textup{a}(\mathcal{U},\mathcal{E}) \leq \frac{d}{d+1}\sqrt{\mathfrak{u}}+\frac{1}{d+1}.
\end{equation}
Note that if \(\mathcal{E}\) is a unitary channel, then both the lower bound and upper bound are saturated.
\end{theorem}
Theorem \ref{theorem-11252053} shows the following applications.
\begin{enumerate}
\item Consider the unitarity certification problem: we want to distinguish that whether \(\mathcal{E}\) is a unitary channel or is \(\epsilon\)-far from any unitary channel in gate infidelity. For the former case, the LHS of Eq.~\eqref{eq-12131554} is saturated (equal to \(1\)); and for the latter case, the LHS of Eq.~\eqref{eq-12131554} is \(\epsilon\)-far from \(1\) since \(\sup_{U\in \mathbb{U}_d} F_\textup{a}(\mathcal{U},\mathcal{E})\leq 1-\epsilon\). Then an estimate \(\hat{\mathfrak{u}}\) for \(\mathfrak{u}(\mathcal{E})\) to precision \(O(\epsilon)\) suffices for this task. This is because, if the estimated value \(\hat{\mathfrak{u}}\) is less than a constant, say \(1/2\), then we can directly claim the the LHS of Eq.~\eqref{eq-12131554} is \(\Omega(1)\)-far from \(1\); otherwise, \(\hat{\mathfrak{u}}\) implies an \(O(\epsilon)\)-precise estimate for the LHS of Eq.~\eqref{eq-12131554}, which is sufficient for distinguishing the two cases.

\item For noise process \(\mathcal{E}\), the RHS of Eq.~\eqref{eq-12131554} provides an upper bound for the best achievable gate fidelity of \(\mathcal{E}\) in the presence of perfect unitary control~\cite{wallman2015estimating}. To see this, we reinterpret the unitary approximability using the identity \(\sup_{U\in\mathbb{U}_d}F_\textup{a}(\mathcal{U},\mathcal{E})=\sup_{V,W\in\mathbb{U}_d}F_\textup{a}(\mathcal{I},\mathcal{V}\mathcal{E}\mathcal{W})\), where \(\mathcal{V}\) and \(\mathcal{W}\) can be seen as the ``recalibrating'' unitary channels for noise reduction.
Note that the upper bound by the RHS of Eq.~\eqref{eq-12131554} is incomparable to that given in~\cite{wallman2015estimating} (for example, consider the cases that \(\mathcal{E}=\mathcal{I}/2\) and \(\mathcal{E}\) is completely depolarizing, see Appendix~\ref{sec-ver} for more details).
\end{enumerate}
We believe further applications are also possible such as the average-to-worst-case error reduction~\cite{kueng2016comparing,wallman2015bounding} and the decomposition for the gate fidelity of composite channels~\cite{carignan2019bounding}.

\section*{Acknowledgment}

The authors thank Wang Fang for helpful discussions about writing. 

This work was supported in part by the National Key Research and Development Program of China under Grant 2018YFA0306701 and in part by the National Natural Science Foundation of China under Grant 61832015. 
Q. Wang was supported by the MEXT Quantum Leap Flagship Program (MEXT Q-LEAP) grants No. JPMXS0120319794.

\bibliographystyle{IEEEtran}
\bibliography{IEEEabrv,bib}

\appendix

\section{SWAP test for partial density operators}\label{sec-01190123}
Since we do not assume the trace-preservation of quantum channels, the output states \(\rho,\sigma\) can be partial density operators, i.e., \(\tr(\rho),\tr(\sigma)\leq 1\). Thus there can be a non-zero probability of obtaining an invalid result, or a failure \(\bot\) after performing the measurement. To tackle this issue, we simply set the current estimation to \(0\) if a failure \(\bot\) occurs. More specifically, suppose \(\rho,\sigma\) are partial density operators such that \(\tr(\rho),\tr(\sigma)\leq 1\), and let \(U_{\textup{ST}}\) denote the unitary corresponding to the circuit of SWAP test. Then, the extended SWAP test for partial density operators is as follows,
\begin{enumerate}
\item apply \(U_{\textup{ST}}\) on \(\ketbra{0}{0}\otimes\rho \otimes \sigma\),
\item measure with the POVM \(\left\{\ketbra{0}{0}\otimes I\otimes I,\ketbra{1}{1}\otimes I\otimes I\right\}\),
\item set \(w=1\) if the measurement outcome is \(0\), \(w=-1\) if the measurement outcome is \(1\), and \(w=0\) if the measurement outcome is a failure \(\bot\).
\end{enumerate}
It is easy to see that \(w\) is an unbiased estimator, i.e., \(\mathbb{E}[w]=\tr(\rho\sigma)\).

\section{Distributed quantum inner product estimation for partial density operators}\label{sec-1130109}
We extend the distributed quantum inner product estimation algorithm~\cite{anshu2022distributed} to accommodate the partial density operators (i.e., the positive operators with trace less than or equal to \(1\)) as input. If a density operator \(\rho\) has trace less than \(1\). Then any measurement on \(\rho\) has a non-zero probability \(1-\tr(\rho)\) to output an invalid result, or a failure \(\bot\). Our algorithm follows essentially the original single-copy measurement DQIPE algorithm~\cite{anshu2022distributed}, but with modified estimators to handle the partial density operators. 

We start by considering the partial collision estimator.
Suppose \(P,Q\in [0,1] \) are two unknown probabilities, and \(p',q'\) are two unknown distributions supported on \(\{0,\ldots,d-1\}\). We are given i.i.d samples \(x_1,\ldots,x_m\) and \(y_1,\ldots,y_m\), where with probability \(P\), \(x_i\sim p'\), and \(x_i=\bot\) otherwise; with probability \(Q\), \(y_i\sim q'\), and \(y_i=\bot\) otherwise. Our goal is to estimate the inner product 
\begin{equation}
g:=PQ\sum_{b=0}^{d-1}p'_b q'_b=\sum_{b=0}^{d-1}p_b q_b,
\end{equation}where \(p,q\) are two partial probability distributions defined as \(p_b:=P p'_b\) and \(q_b:= Q q'_b\).
\begin{definition}[Partial collision estimator]
Given samples \(x_1,\ldots,x_m\) and \(y_1,\ldots,y_m\), where for each \(x_i\), \(y_i\)
\begin{enumerate}
    \item with probability \(P\), \(x_i\) is sampled from the distribution \(p'\), and \(x_i=\bot\) otherwise;
    \item with probability \(Q\), \(y_i\) is sampled from the distribution \(q'\), and \(y_i=\bot\) otherwise.
\end{enumerate}
We define the partial collision estimator as
\begin{equation}\label{eq-11272324}
\tilde{g}=\frac{1}{m^2}\sum_{j,k=1}^m 1[x_j=y_k\neq \bot],
\end{equation}
where \(1[X]=1\) when \(X\) is true and \(0\) otherwise.
\end{definition}
It is easy to see that the partial collision estimator \(\tilde{g}\) is an unbiased estimator for \(g\), i.e., \(\mathbb{E}[\tilde{g}]=g\). Furthermore, its variance can also be bounded. Then, we make use of the partial collision estimator in the distributed quantum inner product estimation for partial density operators, as shown in Algorithm~\ref{alg_einnerproduct}.

\begin{algorithm}[ht]
    \caption{\raggedright Distributed quantum inner product estimation for partial density operators}\label{alg_einnerproduct}
    \begin{algorithmic}[1]
    \Require number of basis settings $N$, number of measurements for each basis $m$, $2Nm$ copies of unknown partial density operators $\rho,\sigma$ acting on $\mathbb{C}^{d}$
    \Ensure an estimate of $\tr(\rho\sigma)$
    \For {$i=1\cdots N$}
        \State sample a random unitary matrix $U_i\sim\mathbb{U}(d)$
        \State measure $m$ copies of $\rho$ in the basis $\{U_i^\dag\ketbra{b}{b}U_i\}_{b=0}^{d-1}$ and obtain $X_i=\{x_1,\dots,x_m\}$
        \State measure $m$ copies of $\sigma$ in the basis $\{U_i^\dag\ketbra{b}{b}U_i\}_{b=0}^{d-1}$ and obtain $Y_i=\{y_1,\dots,y_m\}$
        \State compute the partial collision estimator~\eqref{eq-11272324} using $X_i$ and $Y_i$, denote by $\tilde{g}_i$
        \State let $w_i=(d+1)\tilde{g}_i$
    \EndFor
    \State measure \(Nm\) copies of \(\rho\) with the trivial POVM \(\{I\}\) and obtain \(X=\{x_1,\ldots,x_{Nm}\}\)
    \State measure \(Nm\) copies of \(\sigma\) with the trivial POVM \(\{I\}\) and obtain \(Y=\{y_1,\ldots,y_{Nm}\}\)
    \State let \(t_1=\frac{1}{Nm}\,\#\{x_i\in X\,|\,x_i\neq \bot\}\)
    \State let \(t_2=\frac{1}{Nm}\,\#\{y_i\in Y\,|\,y_i\neq \bot\}\)
    \State\textbf{Return} $w:=\frac{1}{N}\sum_{i=1}^N w_i - t_1t_2$
    \end{algorithmic}
\end{algorithm}

\begin{lemma}\label{lm-1128217}
The output \(w\) of Algorithm~\ref{alg_einnerproduct} is an unbiased estimator for \(\tr(\rho\sigma)\), i.e.,
\begin{equation}
\mathbb{E} [w]=\tr(\rho\sigma).
\end{equation}
\end{lemma}
\begin{proof}
In a single iteration in Algorithm~\ref{alg_einnerproduct}, \(\rho,\sigma\) is rotated by a random unitary \(U_i\), then the inner product of \(p(U_i),q(U_i)\) is estimated, where
\begin{equation}
p_b(U_i)=\bra{b}U_i\rho U_i^\dag\ket{b},\quad\quad q_b(U_i)=\bra{b} U_i\sigma U_i^\dag \ket{b}.
\end{equation}
Note that \(p(U_i),q(U_i)\) are sub-probability distributions (i.e. \(\sum_b p_b(U_i)\leq 1\) and \(\sum_b q_b(U_i)\leq 1\)). Their inner product is:
\begin{equation}
g(U_i)=\sum_{b=0}^{d-1}\bra{b}U_i\rho U_i^\dag\ket{b}\,\bra{b}U_i\sigma U_i^\dag\ket{b}.
\end{equation}
Let \(\tilde{g}_i\) denote the collision estimator (line 5 of Algorithm~\ref{alg_einnerproduct}), with unitary \(U_i\) and samples \(S_i=\{X_i,Y_i\}\). Its expectation is given by
\begin{equation}
\begin{split}
\mathbb{E}_{U_i,S_i}\,[\tilde{g}_i]&=\mathbb{E}_{U_i}\,[g(U_i)]=d \, \mathbb{E}_{\psi}\, \bra{\psi}\rho\ket{\psi}\,\bra{\psi}\sigma\ket{\psi}\\
&=\frac{1}{d+1}\tr\left((I\otimes I+\textup{SWAP})\rho\otimes \sigma\right)\\
&=\frac{\tr(\rho)\tr(\sigma)+\tr(\rho\sigma)}{d+1}.
\end{split}
\end{equation}
Thus \(\mathbb{E} [w_i] = \tr(\rho)\tr(\sigma)+\tr(\rho\sigma)\).
On the other hand, it is easy to see that \(\mathbb{E}[t_1]=\tr(\rho)\) and \(\mathbb{E}[t_2]=\tr(\sigma)\).
Since \(t_1, t_2\) are independent, we have \(\mathbb{E} [t_1t_2]=\tr(\rho)\tr(\sigma)\). Then,
\begin{equation}
\mathbb{E}[w]= \tr(\rho)\tr(\sigma)+\tr(\rho\sigma)-\tr(\rho)\tr(\sigma)=\tr(\rho\sigma).
\end{equation}
\end{proof}
For the variance of \(\tilde{g}_i\), we have the following lemma.
\begin{lemma}\label{lm-1128215}
The total variance of the partial collision estimator \(\tilde{g}_i\) (line 5 of Algorithm~\ref{alg_einnerproduct}) is upper bounded by
\begin{equation}
\Var\left(\tilde{g}_i\right)=O\left(\frac{1}{d^3}+\frac{1}{m^2d}+\frac{1}{md^2}\right).
\end{equation}
\end{lemma}
Note that the proof is essentially the same as that in \cite{anshu2022distributed}, except all formulas are made homogeneous with respect to \(\rho\) and \(\sigma\). For completeness, we provide a proof in the Appendix. On the other hand, it is easy to see that the variance of \(t_1t_2\) is also bounded.
\begin{lemma}\label{lm-1128216}
The variance of \(t_1t_2\) is upper bounded by
\begin{equation}
\Var(t_1t_2)=O\left(\frac{1}{Nm}\right).
\end{equation}
\end{lemma}
Then, combining the previous results, we have the following proposition.
\begin{proposition}\label{tm-1128230}
The output \(w\) of Algorithm \ref{alg_einnerproduct} is an unbiased estimator for \(\tr(\rho\sigma)\), and its variance is upper bounded by
\begin{equation}
\Var(w)=O\left(\frac{1}{Nd}+\frac{d}{N m^2}+\frac{1}{Nm}\right).
\end{equation}
\end{proposition}
\begin{proof}
Since \(w_i\) and \(t_1t_2\) are independent random variables, we have
\begin{equation}
\begin{split}
\Var(w)&=\frac{1}{N}\Var(w_i)+\Var(t_1t_2)=\frac{(d+1)^2}{N}\Var(\tilde{g}_i)+\Var(t_1t_2)\\
&=O\left(\frac{1}{Nd}+\frac{d}{Nm^2}+\frac{1}{Nm}\right)+O\left(\frac{1}{Nm}\right)=O\left(\frac{1}{Nd}+\frac{d}{Nm^2}+\frac{1}{Nm}\right).
\end{split}
\end{equation}
\end{proof}

\subsection{Proof of Lemma~\ref{lm-1128215}}
For completeness, we provide a proof of Lemma~\ref{lm-1128215} in this section. The proof is essentially the same as that in \cite{anshu2022distributed}, except all formulas are made homogeneous with respect to \(\rho\) and \(\sigma\).
\begin{lemma}[See, for example, Lemma 22 in \cite{anshu2022distributed}]\label{lemma-1252357}
Let $A,B,C,D$ be Hermitian matrices. Then
\begin{equation}
    \mathbb{E}_{\psi\sim\mathbb{C}^d}\bra{\psi}A\ket{\psi}\bra{\psi}B\ket{\psi}=\frac{1}{d(d+1)}\left(\tr(A)\tr(B)+\tr(AB)\right),
\end{equation}
and
\begin{equation}
\begin{split}
    \mathbb{E}_{\psi\sim\mathbb{C}^d}\bra{\psi}A\ket{\psi}\bra{\psi}B\ket{\psi}\bra{\psi}C\ket{\psi}&=\frac{1}{d(d+1)(d+2)}\bigg(\tr(A)\tr(B)\tr(C)+\tr(AB)\tr(C)\\
    &+\tr(A)\tr(BC)+\tr(B)\tr(AC)+\tr(ABC)+\tr(ACB)\bigg).
\end{split}
\end{equation}
Similarly, we have:
\begin{equation}
\begin{split}
    &\mathbb{E}_{\psi\sim\mathbb{C}^d}\bra{\psi}A\ket{\psi}\bra{\psi}B\ket{\psi}\bra{\psi}C\ket{\psi}\bra{\psi}D\ket{\psi}\\
    &=\frac{1}{d(d+1)(d+2)(d+3)}\sum_{\pi\in S_4}\tr(A\otimes B\otimes C\otimes D\cdot  P_d(\pi)),
\end{split}
\end{equation}
\end{lemma}

\begin{lemma}\label{lm-1252329}
The variance of the partial collision estimator \eqref{eq-11272324} is upper bounded by
\begin{equation}
\Var(\tilde{g})\leq \frac{g}{m^2}+\frac{1}{m}\left(\sum_{b=0}^{d-1}p_b q^2_b+\sum_{b=0}^{d-1}p^2_b q_b\right).
\end{equation}
\end{lemma}
\begin{proof}
\begin{equation}
\begin{split}
\Var(\tilde{g})&=\mathbb{E}\tilde{g}^2-g^2=\frac{1}{m^4}\sum_{j,k,l,s}\mathbb{E}\,1[x_j=y_k\neq\bot]\cdot 1[x_l=y_s\neq\bot]-g^2\\
&=\frac{1}{m^4}\left(m^2g+m^2(m-1)^2g^2+m^2(m-1)PQ^2\sum_{b=0}^{d-1}p'_bq'^2_b+m^2(m-1)P^2Q\sum_{b=0}^{d-1}p'^2_bq'_b\right)-g^2\\
&\leq \frac{g}{m^2}+\frac{1}{m}\left(PQ^2\sum_{b=0}^{d-1}p'_b q'^2_b+P^2Q\sum_{b=0}^{d-1}p'^2_b q'_b\right)
=\frac{g}{m^2}+\frac{1}{m}\left(\sum_{b=0}^{d-1}p_b q^2_b+\sum_{b=0}^{d-1}p^2_b q_b\right).
\end{split}
\end{equation}
In the third line, the tour terms in the bracket come from \(4\) different cases in the sum: \(j=l\wedge k=s\), \(j\neq l \wedge k\neq s\), \(j=l\wedge k\neq s\) and \(j\neq l\wedge k=s\), respectively.
\end{proof}

Now, we consider the variance of \(\tilde{g}_i\). For convenience, we omit the subscript \(i\). Then, by the law of total variance, we have
\begin{equation}\label{eq-126023}
\Var(\tilde{g})=\Var_{U}(\mathbb{E}[\tilde{g}|U])+\mathbb{E}_{U}[\Var(\tilde{g}|U)].
\end{equation}
For the second term of Eq.~\eqref{eq-126023}, using the expression for the variance of the partial collision estimator in Lemma~\ref{lm-1252329}, we have
\begin{equation}
\mathbb{E}_{U}[\Var(\tilde{g}|U)]\leq \mathbb{E}_{U}\left[\frac{g(U)}{m^2}+\frac{1}{m}\left(\sum_{b}p_b(U) q_b(U)^2+\sum_{b}p_b(U)^2q_b(U)\right)\right].
\end{equation}
It was shown in Lemma~\ref{lm-1128217} that \(\mathbb{E}_{U}[g(U)]=\frac{\tr(\rho)\tr(\sigma)+\tr(\rho\sigma)}{d+1}\). And we have
\begin{equation}
\begin{split}
\mathbb{E}_{U}\sum_{b=0}^{d-1} p_b(U)q_b(U)^2&=\mathbb{E}_{U}\sum_{b=0}^{d-1} \,\bra{b}U\rho U^\dag\ket{b} \, \bra{b}U \sigma U^\dag \ket{b}^2\\
&=d\,\mathbb{E}_{\psi\sim\mathbb{C}^d}\,\bra{\psi}\rho\ket{\psi}\,\bra{\psi}\sigma\ket{\psi}^2\\
&=O(1/d^2),
\end{split}
\end{equation}
where the third line follows from Lemma~\ref{lemma-1252357}. Thus we have
\begin{equation}
\mathbb{E}_{U}\left[\Var(\tilde{g}|U)\right]=O\left(\frac{1}{m^2 d}+\frac{1}{md^2}\right).
\end{equation}
Then, for the first term of Eq.~\eqref{eq-126023}, we have
\begin{equation}\label{eq-126237}
\begin{split}
\Var_{U}\left(\mathbb{E}[\tilde{g}|U]\right)&=\Var_{U}\left(g(U)\right)=\mathbb{E}_{U}[g(U)^2]-\frac{\left(\tr(\rho)\tr(\sigma)+\tr(\rho\sigma)\right)^2}{(d+1)^2}\\
&=\sum_{b_1,b_2=0}^{d-1}\mathbb{E}_{U}\left[\bra{b_1} U\rho U^\dag\ket{b_1}\bra{b_1}U\sigma U^\dag\ket{b_1}\bra{b_2} U\rho U^\dag\ket{b_2}\bra{b_2}U\sigma U^\dag\ket{b_2}\right]\\
&\quad\quad-\frac{\left(\tr(\rho)\tr(\sigma)+\tr(\rho\sigma)\right)^2}{(d+1)^2}\\
&=d\,\mathbb{E}_{\psi\sim \mathbb{C}^d}\left[\bra{\psi}\rho\ket{\psi}^2\,\bra{\psi}\sigma\ket{\psi}^2\right]
+d(d-1)\mathbb{E}_{\substack{\psi\sim \mathbb{C}^d\\\psi'\sim\mathbb{C}^{d-1}_{\psi_\perp}}}\bra{\psi}\rho\ket{\psi}\bra{\psi}\sigma\ket{\psi}\bra{\psi'}\rho\ket{\psi'}\bra{\psi'}\sigma\ket{\psi'}\\ &\quad\quad-\frac{\left(\tr(\rho)\tr(\sigma)+\tr(\rho\sigma)\right)^2}{(d+1)^2},
\end{split}
\end{equation}
where \(\psi'\sim\mathbb{C}^{d-1}_{\psi_\perp}\) means \(\ket{\psi'}\) is a uniformly random vector in the \(d-1\)-dimensional subspace that is perpendicular to \(\ket{\psi}\). Then, we have
\begin{equation}\label{eq-126218}
\begin{split}
&\mathbb{E}_{\psi\sim \mathbb{C}^d,\psi'\sim\mathbb{C}^{d-1}_{\psi_\perp}}\bra{\psi}\rho\ket{\psi}\bra{\psi}\sigma\ket{\psi}\bra{\psi'}\rho\ket{\psi'}\bra{\psi'}\sigma\ket{\psi'}\\
=&\mathbb{E}_{\psi\sim\mathbb{C}^d}\left[\bra{\psi}\rho\ket{\psi}\bra{\psi}\sigma\ket{\psi}\mathbb{E}_{\psi'\sim \mathbb{C}_{\psi_\perp}^{d-1}}\left[\bra{\psi'}\rho\ket{\psi'}\bra{\psi'}\sigma\ket{\psi'}\right]\right].
\end{split}
\end{equation}
Then, we have the following result (see, for example, Lemma 18 in \cite{anshu2022distributed}),
\begin{equation}
\mathbb{E}_{\psi'\sim\mathbb{C}_{\psi_\perp}^{d-1}}\ketbra{\psi'}{\psi'}^{\otimes 2}=\frac{1}{d(d-1)}\left((I-\ketbra{\psi}{\psi})^{\otimes 2}+(I-\ketbra{\psi}{\psi})^{\otimes 2}\textup{SWAP}(I-\ketbra{\psi}{\psi})^{\otimes 2}\right).
\end{equation}
Therefore, we have
\begin{equation}
\begin{split}
&d(d-1)\mathbb{E}_{\psi'\sim\mathbb{C}_{\psi_\perp}^{d-1}}[\bra{\psi'}\rho\ket{\psi'}\bra{\psi'}\sigma\ket{\psi'}]\\
&=(\tr(\rho)-\bra{\psi}\rho\ket{\psi})(\tr(\sigma)-\bra{\psi}\sigma\ket{\psi})+\tr(\rho(I-\ketbra{\psi}{\psi})\sigma(I-\ketbra{\psi}{\psi}))\\
&=(\tr(\rho)-\bra{\psi}\rho\ket{\psi})(\tr(\sigma)-\bra{\psi}\sigma\ket{\psi})+\tr(\rho\sigma)+\bra{\psi}\rho\ket{\psi}\bra{\psi}\sigma\ket{\psi}-\bra{\psi}(\rho\sigma+\sigma\rho)\ket{\psi}\\
&=\tr(\rho)\tr(\sigma)+\tr(\rho\sigma)+2\bra{\psi}\rho\ket{\psi}\bra{\psi}\sigma\ket{\psi}-\tr(\rho)\bra{\psi}\sigma\ket{\psi}-\tr(\sigma)\bra{\psi}\rho\ket{\psi}-\bra{\psi}(\rho\sigma+\sigma\rho)\ket{\psi}.
\end{split}
\end{equation}
Plugging this into Eq.~\eqref{eq-126218}, we have
\begin{equation}
\begin{split}
&d(d-1)\mathbb{E}_{\psi\sim \mathbb{C}^d,\psi'\sim\mathbb{C}^{d-1}_{\psi_\perp}}\bra{\psi}\rho\ket{\psi}\bra{\psi}\sigma\ket{\psi}\bra{\psi'}\rho\ket{\psi'}\bra{\psi'}\sigma\ket{\psi'}\\
&=\mathbb{E}_{\psi\sim\mathbb{C}^d}\Big[(\tr(\rho)\tr(\sigma)+\tr(\rho\sigma))\bra{\psi}\rho\ket{\psi}\bra{\psi}\sigma\ket{\psi}+2\bra{\psi}\rho\ket{\psi}^2\bra{\psi}\sigma\ket{\psi}^2\\
&-\tr(\rho)\bra{\psi}\rho\ket{\psi}\bra{\psi}\sigma\ket{\psi}^2- \tr(\sigma)\bra{\psi}\rho\ket{\psi}^2\bra{\psi}\sigma\ket{\psi}-\bra{\psi}\rho\ket{\psi}\bra{\psi}\sigma\ket{\psi}\bra{\psi}(\rho\sigma+\sigma\rho)\ket{\psi} \Big]\\
&=\frac{(\tr(\rho)\tr(\sigma)+\tr(\rho\sigma))^2}{d(d+1)}+O\left(\frac{1}{d^4}\right),
\end{split}
\end{equation}
where the the negative terms are ignored (the positivity of \(\bra{\psi}\rho\ket{\psi}\bra{\psi}\sigma\ket{\psi}\bra{\psi}(\rho\sigma+\sigma\rho)\ket{\psi}\) can be verified by Lemma~\ref{lemma-1252357}). Then Eq.~\eqref{eq-126237} becomes
\begin{equation}
\begin{split}
\Var_{U}\left(\mathbb{E}[\tilde{g}|U]\right)
&=d\cdot O\left(\frac{1}{d^4}\right)+\frac{(\tr(\rho)\tr(\sigma)+\tr(\rho\sigma))^2}{d(d+1)}+O\left(\frac{1}{d^4}\right)-\frac{\left(\tr(\rho)\tr(\sigma)+\tr(\rho\sigma)\right)^2}{(d+1)^2}\\
&=O\left(\frac{1}{d^3}\right),
\end{split}
\end{equation}
and the total variance is given by
\begin{equation}
\Var(\tilde{g})=\Var_U(\mathbb{E}[\tilde{g}|U])+\mathbb{E}_U[\Var(\tilde{g}|U)]=O\left(\frac{1}{d^3}+\frac{1}{m^2d}+\frac{1}{md^2}\right).
\end{equation}

\section{Benchmarking quantum processes with unitarity}\label{sec-01191541}
\begin{theorem}\label{theorem-11252053-2}
Suppose \(\mathcal{E}\) is a quantum channel acting on a \(d\)-dimensional system. For convenience, let \(\mathfrak{u}\) stands for \(\mathfrak{u}(\mathcal{E})\). Then,
\begin{equation}\label{eq-12131554-2}
\frac{d\mathfrak{u}^2+\sqrt{\mathfrak{u}}}{d+1}\leq \sup_{U\in \mathbb{U}_d} F_\textup{a}(\mathcal{U},\mathcal{E}) \leq \frac{d\sqrt{\mathfrak{u}}+1}{d+1}.
\end{equation}
Note that if \(\mathcal{E}\) is a unitary channel, then both the lower bound and upper bound are saturated.
\end{theorem}
\begin{proof}
Since the quantum channel \(\mathcal{E}\) can be non-trace-preserving, for convenience, we define the following trace-preservation index:
\begin{equation}\label{eq-11262004}
\mathfrak{t}(\mathcal{E}):=\tr\left(\mathcal{E}(I/d)\right),
\end{equation}
which is a quantification of the trace-preservation property of the quantum channel \(\mathcal{E}\).
\begin{lemma}\label{lm-12131618}
The trace-preservation index \(\mathfrak{t}(\mathcal{E})\) has the following properties:
\begin{enumerate}
\item \(\mathfrak{t}(\mathcal{E})\leq 1\) with equality if and only if \(\mathcal{E}\) is trace-preserving,
\item \(\mathfrak{t}(\mathcal{E})^2\geq \mathfrak{u}(\mathcal{E})\),
\item \(\mathfrak{t}(\mathcal{E})\) can be easily estimated to precision \(O(\epsilon)\) with \(O(\epsilon^{-2})\) incoherent calls to \(\mathcal{E}\), by simply inputting the maximally mixed state \(I/d\) and measuring with the trivial POVM \(\{I\}\).
\end{enumerate}
\end{lemma}
\begin{proof}
The first and third properties are obvious. For the second property, we have
\begin{equation}
\mathfrak{t}(\mathcal{E})^2=\frac{1}{d^2}\sum_{i,j}\tr(E_i^\dag E_i)\tr(E_j^\dag E_j)\geq \frac{1}{d^2}\sum_{i,j}|\tr(E_i^\dag E_j)|^2=\mathfrak{u}(\mathcal{E}),
\end{equation}
where \(E_i\) are the Kraus operators of quantum channel \(\mathcal{E}\).
\end{proof}
For convenience, we may directly use \(\mathfrak{t}\) to denote \(\mathfrak{t}(\mathcal{E})\) in this proof.
First, consider the right inequality in Eq.~\eqref{eq-12131554}. Note that the average fidelity can be written as
\begin{equation}
F_\textup{a}(\mathcal{U},\mathcal{E})=\frac{\tr\left(\mathcal{U}^\dag \mathcal{E}\right)+\tr\left(\mathcal{E}(I)\right)}{d(d+1)}.
\end{equation}
This is in fact the relation between average gate fidelity and entanglement fidelity (see, for example, Proposition 42 in \cite{kliesch2021theory}).
By Cauchy's inequality
\begin{equation}
\tr\left(\mathcal{U}^\dag\mathcal{E}\right)\leq \sqrt{\tr\left(\mathcal{U}^\dag\mathcal{U}\right)\tr\left(\mathcal{E}^\dag\mathcal{E}\right)}=\sqrt{d^2\,\,\tr\left(\mathcal{E}^\dag\mathcal{E}\right)}=d^2\sqrt{\mathfrak{u}}.
\end{equation}
Thus, for any unitary \(U\),
\begin{equation}
F_\textup{a}(\mathcal{U},\mathcal{E})\leq \frac{d^2\sqrt{\mathfrak{u}(\mathcal{E})}+\tr(\mathcal{E}(I))}{d(d+1)}=\frac{d}{d+1}\sqrt{\mathfrak{u}}+\frac{1}{d+1}\mathfrak{t}.
\end{equation}
Therefore,
\begin{equation}
\sup_{U\in\mathbb{U}_d} F_\textup{a}(\mathcal{U},\mathcal{E})\leq \frac{d}{d+1}\sqrt{\mathfrak{u}}+\frac{1}{d+1}\mathfrak{t}\leq \frac{d}{d+1}\sqrt{\mathfrak{u}}+\frac{1}{d+1}.
\end{equation}

Then, consider the left inequality in Eq.~\eqref{eq-12131554}. We choose a Kraus representation \(\sum_{i=1}^n E_i\circ E_i^\dag\) of \(\mathcal{E}\) such that \(\tr(E_i^\dag E_j)=0\) for all \(i\neq j\). This is always possible due to the unitary freedom in the Kraus representation (see, for example, Theorem 8.2 in \cite{nielsen2002quantum}). Thus,
\begin{equation}
\begin{split}
d^2\,\mathfrak{u}=\tr\left(\mathcal{E}^\dag\mathcal{E}\right)=\sum_{ij}\tr\left[E^\dag_i\otimes E_i^{\dag *} \cdot E_j \otimes E_j^*\right]
=\sum_{ij}\left|\tr\left(E_i^\dag E_j\right)\right|^2=\sum_{i}\tr\left(E_i^\dag E_i\right)^2,
\end{split}
\end{equation}
and
\begin{equation}
d\,\mathfrak{t}=\sum_i \tr\left(E_i^\dag E_i\right).
\end{equation}
Without loss of generality, we assume that \(\tr(E_1^\dag E_1)\geq \cdots\geq \tr(E_n^\dag E_n)\). Suppose the singular value decomposition of \(E_1\) is \(U_1 D_1 V_1^\dag\). Then
\begin{equation}\label{eq11252032}
\begin{split}
\sup_{U\in\mathbb{U}_d} F_\textup{a}(\mathcal{U},\mathcal{E})&\geq F_\textup{a}(\mathcal{U}_1\mathcal{V}_1^\dag,\mathcal{E})=\frac{\tr\left(\mathcal{V}_1 \mathcal{U}_1^\dag\mathcal{E}\right)+d\,\mathfrak{t}}{d(d+1)} = \frac{\sum_i \left|\tr\left(V_1U_1^\dag E_i\right)\right|^2+d\,\mathfrak{t}}{d(d+1)}\\
&\geq\frac{\left|\tr\left(V_1U^\dag_1E_1\right)\right|^2+d\,\mathfrak{t}}{d(d+1)}
=\frac{\tr\left(\sqrt{E_1^\dag E_1}\right)^2+d\, \mathfrak{t}}{d(d+1)}\\
&\geq\frac{\tr\left(E_1^\dag E_1\right)^2+d\, \mathfrak{t}}{d(d+1)},
\end{split}
\end{equation}
where the last inequality is because that \(E_1^\dag E_1\sqsubseteq I\) and thus its eigenvalues are smaller than \(1\). Then, we have
\begin{equation}
d^2\, \mathfrak{u}=\sum_i\tr\left(E_i^\dag E_i\right)^2\leq \tr\left(E_1^\dag E_1\right)\sum_i\tr\left(E_i^\dag E_i\right)=\tr\left(E_1^\dag E_1\right) d\,\mathfrak{t},
\end{equation}
which means
\begin{equation}
\tr\left(E_1^\dag E_1\right)\geq d\,\frac{\mathfrak{u}}{\mathfrak{t}}.
\end{equation}
Thus
\begin{equation}
\eqref{eq11252032}\geq\frac{d^2\frac{\mathfrak{u}^2}{\mathfrak{t}^2}+d\,\mathfrak{t}}{d(d+1)}=\frac{d}{d+1}\frac{\mathfrak{u}^2}{\mathfrak{t}^2}+\frac{1}{d+1}\mathfrak{t}\geq
\frac{d}{d+1}\mathfrak{u}^2+\frac{1}{d+1}\sqrt{\mathfrak{u}},
\end{equation}
where the second inequality is due to the first two properties in Lemma~\ref{lm-12131618}.
\end{proof}

\section{Alternative definition of unitarity}\label{sec-ver}
There is another definition of unitarity proposed in \cite{wallman2015estimating}, which is to characterize the coherence of noise processes. We will use \(\mathfrak{u}'\) to denote this alternative definition:
\begin{equation}
\mathfrak{u}' (\mathcal{E}):=\frac{d}{d-1} \int_\psi \textup{d}\psi \,\,\tr\left[\mathcal{E}'(\psi)^\dag \mathcal{E}'(\psi)\right],
\end{equation}
where \(\mathcal{E}'\) is a linear map such that \(\mathcal{E}'(X)=\mathcal{E}(X)-\left[\tr[\mathcal{E}(X)]/d\right]I\) for all traceless \(X\) and \(\mathcal{E}'(I)=0\). The definition of \(\mathcal{E}'\) aims to eliminate the effects caused by the non-unitality and state-dependent leakage of quantum channel \(\mathcal{E}\). We will show that our upper bounds and lower bounds also apply to this alternative unitarity (see Table~\ref{table-12132342}).
\begin{table}[ht]
\centering
\renewcommand{\arraystretch}{1.5}
\setlength{\tabcolsep}{3mm}{
\begin{tabular}{ccc}
\hline
& Coherent access & Incoherent access \\
\hline
Upper bound & \(O(\epsilon^{-2})\) &   \(O(\sqrt{d}\cdot\epsilon^{-2})\)  \\
Lower bound & \(\Omega(\epsilon^{-2})\) &  \(\Omega(\sqrt{d}+\epsilon^{-2})\) \\
\hline
\end{tabular}
}
\caption{Our results for the alternative unitarity \(\mathfrak{u}'\).}
\label{table-12132342}
\end{table}

As a quick start, we remark that (see Section~\ref{sec-12162325})
\begin{equation}
|\mathfrak{u}(\mathcal{E})-\mathfrak{u}'(\mathcal{E})|\leq O(d^{-1}).
\end{equation}
Thus for the case that \(\epsilon\geq d^{-1}\), an \(O(\epsilon)\)-estimate for \(\mathfrak{u}\) directly implies an \(O(\epsilon)\)-estimate for \(\mathfrak{u}'\) and vice versa, which further implies that the upper and lower bounds for \(\mathfrak{u}\) apply to \(\mathfrak{u}'\). For the general case, we can still obtain the same results, but with more sophisticated strategy (see Section~\ref{sec-12162320} and Section~\ref{sec-12162322}).

\subsection{Closeness}\label{sec-12162325}
Suppose \(\mathcal{E}\) is a quantum channel acting on a \(d\)-dimensional system. Consider the following subspaces of \(\mathbb{C}^{d^2}\):
\begin{equation}\label{eq-11261629}
\{\kett{cI}\,|\, c\in\mathbb{C}\},\quad\quad \{\kett{X}\,|\,\tr(X)=0, X\in\mathbb{M}_d\}.
\end{equation}
Then, the matrix representation of quantum channel \(\mathcal{E}\), under the orthonormal basis vectors from these two subspaces respectively, can be written as
\begin{equation}
\begin{bmatrix}
\mathfrak{t}(\mathcal{E}) & \mathcal{E}_{\textup{sdl}}\\
\mathcal{E}_{\textup{n}} & \mathcal{E}_{\textup{u}}
\end{bmatrix},
\end{equation}
where \(\mathfrak{t}(\mathcal{E})\) is the trace-preservation index defined in Eq. \eqref{eq-11262004}, \(\mathcal{E}_{\textup{sdl}}\), \(\mathcal{E}_{\textup{n}}\) and \(\mathcal{E}_{\textup{u}}\) refer to the state-dependent leakage, nonunital and unital blocks of \(\mathcal{E}\), respectively. Then, we have
\begin{proposition}[See \cite{wallman2015estimating}]
\begin{equation}\label{eq-11261702}
\mathfrak{u}'(\mathcal{E})=\frac{1}{d^2-1}\tr\left(\mathcal{E}_{\textup{u}}^\dag\mathcal{E}_{\textup{u}}\right).
\end{equation}
\end{proposition}
We also have: 
\begin{proposition}
\begin{equation}\label{eq-11261703}
\mathfrak{u}(\mathcal{E})=\frac{1}{d^2}\left[\mathfrak{t}(\mathcal{E})^2+\tr\left(\mathcal{E}_{\textup{sdl}}^\dag\mathcal{E}_{\textup{sdl}}\right)+\tr\left(\mathcal{E}_{\textup{n}}^\dag\mathcal{E}_{\textup{n}}\right)+\tr\left(\mathcal{E}_{\textup{u}}^\dag\mathcal{E}_{\textup{u}}\right)\right].
\end{equation}
\end{proposition}
Then, we can bound the difference between \(\mathfrak{u}\) and \(\mathfrak{u}'\) using Eq. \eqref{eq-11261702} and Eq. \eqref{eq-11261703}.
\begin{proposition}[Closeness of \(\mathfrak{u}\) and \(\mathfrak{u}'\)]
For any quantum channel \(\mathcal{E}\),
\begin{equation}
\left|\mathfrak{u}(\mathcal{E})-\mathfrak{u}'(\mathcal{E})\right|\leq \frac{1}{d}+\frac{1}{d^2}= O(\frac{1}{d}).
\end{equation}
\end{proposition}
\begin{proof}
\begin{equation}
\begin{split}
d&\geq d\,\tr\bigl[\mathcal{E}\left(I/d\right)\mathcal{E}\left(I/d\right)\bigr]=\tr\left[\frac{I}{\sqrt{d}}\mathcal{E}^\dag\mathcal{E}\left(\frac{I}{\sqrt{d}}\right)\right]\\
&=\begin{bmatrix}1 & 0\end{bmatrix}
\begin{bmatrix}\mathfrak{t}(\mathcal{E})^2+\mathcal{E}_{\textup{n}}^\dag\mathcal{E}_{\textup{n}}&\cdot\\\cdot&\mathcal{E}_{\textup{sdl}}^\dag\mathcal{E}_{\textup{sdl}}+\mathcal{E}_{\textup{u}}^\dag\mathcal{E}_{\textup{u}}\end{bmatrix}
\begin{bmatrix}1\\0\end{bmatrix}=\mathfrak{t}(\mathcal{E})^2+\mathcal{E}_\textup{n}^\dag\mathcal{E}_{\textup{n}},
\end{split}
\end{equation}
and
\begin{equation}
\begin{split}
d&=\tr\left(I^2\right) \geq\tr\left[\left(\sum_i E_i^\dag E_i\right)^2\right]=\tr\left[\mathcal{E}^\dag(I)\mathcal{E}^\dag(I)\right]=d\tr\left[\frac{I}{\sqrt{d}}\mathcal{E}\mathcal{E}^\dag \left(\frac{I}{\sqrt{d}}\right)\right]\\
&=d\begin{bmatrix}1 & 0\end{bmatrix}
\begin{bmatrix}\mathfrak{t}(\mathcal{E})^2+\mathcal{E}_{\textup{sdl}}\mathcal{E}_{\textup{sdl}}^\dag&\cdot\\\cdot&\mathcal{E}_{\textup{n}}\mathcal{E}_{\textup{n}}^\dag+\mathcal{E}_{\textup{u}}\mathcal{E}_{\textup{u}}^\dag\end{bmatrix}
\begin{bmatrix}1\\0\end{bmatrix}
=d \left(\mathfrak{t}(\mathcal{E})^2+\mathcal{E}_{\textup{sdl}}\mathcal{E}_{\textup{sdl}}^\dag\right).
\end{split}
\end{equation}
Therefore,
\begin{equation}
\frac{1}{d^2}\left[\mathfrak{t}(\mathcal{E})^2+\tr\left(\mathcal{E}_\textup{n}^\dag\mathcal{E}_{\textup{n}}\right)\right]\leq \frac{1}{d},\quad\quad
\frac{1}{d^2}\left[\mathfrak{t}(\mathcal{E})^2+\tr\left(\mathcal{E}_{\textup{sdl}}^\dag\mathcal{E}_{\textup{sdl}}\right)\right]\leq \frac{1}{d^2},
\end{equation}
which directly imply
\begin{equation}
\begin{split}
\left|\mathfrak{u}(\mathcal{E})-\mathfrak{u}'(\mathcal{E})\right|&\leq\left|-\frac{1}{d^2}\mathfrak{u}'(\mathcal{E})+\frac{1}{d^2}\mathfrak{t}(\mathcal{E})^2+\frac{1}{d^2}\tr\left(\mathcal{E}^\dag_{\textup{sdl}}\mathcal{E}_{\textup{sdl}}\right)+\frac{1}{d^2}\tr\left(\mathcal{E}_{\textup{n}}^\dag\mathcal{E}_{\textup{n}}\right)\right|\\
&\leq \frac{1}{d}+\frac{1}{d^2}= O(\frac{1}{d}).
\end{split}
\end{equation}
\end{proof}

\subsection{Upper bound}\label{sec-12162320}
For any unitary channel \(\mathcal{U}\), its matrix representation under the orthonormal basis vectors of the subspaces in \eqref{eq-11261629} can be written as
\begin{equation}
\begin{bmatrix}
1&0\\
0& \mathcal{U}_\textup{u}
\end{bmatrix},
\end{equation}
where \(\{\mathcal{U}_\textup{u}\}_{U\in\mathbb{U}_d}\) is an irreducible representation of the unitary group \(\mathbb{U}_d\). For any pure state \(\rho\), its vector representation under the orthonormal basis vectors of the subspaces in \eqref{eq-11261629} can be written as
\begin{equation}
\begin{bmatrix}
1/\sqrt{d}\\ \sqrt{1-1/d}\bm{\rho}'
\end{bmatrix},
\end{equation}
where \(\bm{\rho}'\) is a unit vector of \(d^2-1\) dimension.

Then, we have
\begin{equation}
\begin{split}
\mathbb{E}_{U}\left[\mathcal{E}^\dag\mathcal{U}\mathcal{E}\right]
&=\begin{bmatrix}
\mathfrak{t}(\mathcal{E})& \mathcal{E}_{\textup{n}}^\dag\\
\mathcal{E}_{\textup{sdl}}^\dag & \mathcal{E}_{\textup{u}}^\dag 
\end{bmatrix}
\begin{bmatrix}
\mathfrak{t}(\mathcal{E})& \mathcal{E}_{\textup{sdl}}\\
\mathbb{E}_{U}\left[\mathcal{U}_{\textup{u}}\right]\mathcal{E}_{\textup{n}}& \mathbb{E}_{U}\left[\mathcal{U}_{\textup{u}}\right]\mathcal{E}_{\textup{u}}
\end{bmatrix}\\
&=\begin{bmatrix}
\mathfrak{t}(\mathcal{E})& \mathcal{E}_{\textup{n}}^\dag\\
\mathcal{E}_{\textup{sdl}}^\dag & \mathcal{E}_{\textup{u}}^\dag 
\end{bmatrix}
\begin{bmatrix}
\mathfrak{t}(\mathcal{E})& \mathcal{E}_{\textup{sdl}}\\
0& 0
\end{bmatrix}
=\begin{bmatrix}
\mathfrak{t}(\mathcal{E})^2 & * \\
* & \mathcal{E}_{\textup{sdl}}^\dag\mathcal{E}_{\textup{sdl}}
\end{bmatrix}.
\end{split}
\end{equation}
Thus
\begin{equation}
\mathbb{E}_{U,V}\left[\mathcal{V}^\dag\mathcal{E}^\dag\mathcal{U}\mathcal{E}\mathcal{V}\right]
=\begin{bmatrix}
\mathfrak{t}(\mathcal{E})^2 & *\,\mathbb{E}_{V}\left[\mathcal{V}_{\textup{u}}\right] \\
\mathbb{E}_{V}\left[\mathcal{V}_{\textup{u}}^\dag\right]*\,\, & \mathbb{E}_{V}\left[\mathcal{V}_{\textup{u}}^\dag\mathcal{E}_{\textup{sdl}}^\dag\mathcal{E}_{\textup{sdl}}\mathcal{V}_{\textup{u}}\right]
\end{bmatrix}
=\begin{bmatrix}
\mathfrak{t}(\mathcal{E})^2& 0\\
0 &  \frac{\tr\left(\mathcal{E}_{\textup{sdl}}^\dag\mathcal{E}_{\textup{sdl}}\right)}{d^2-1} I
\end{bmatrix}.
\end{equation}
Then, we define the quantity \(\mathfrak{s}(\mathcal{E})\) by
\begin{equation}\label{eq-2001102}
\begin{split}
\mathfrak{s}(\mathcal{E}):&=\mathbb{E}_{U,V}\left[\bbra{\rho_0}\mathcal{V}^\dag\mathcal{E}^\dag\mathcal{U}\mathcal{E}\mathcal{V}\kett{\rho_0}\right]\\
&=\begin{bmatrix}1/\sqrt{d}&\sqrt{1-1/d}\bm{\rho}'^{\dag}_0\end{bmatrix}\begin{bmatrix}
\mathfrak{t}(\mathcal{E})^2& 0\\
0 &  \frac{\tr\left(\mathcal{E}_{\textup{sdl}}^\dag\mathcal{E}_{\textup{sdl}}\right)}{d^2-1} I
\end{bmatrix}\begin{bmatrix}1/\sqrt{d}\\\sqrt{1-1/d}\bm{\rho}'_0\end{bmatrix}\\
&=\frac{1}{d}\mathfrak{t}(\mathcal{E})^2+\frac{1}{d(d+1)}\tr\left(\mathcal{E}_{\textup{sdl}}^\dag\mathcal{E}_{\textup{sdl}}\right),
\end{split}
\end{equation}
where \(\rho_0=\ketbra{0}{0}\).
On the other hand, recall the definition of the purity-preservation index \(\mathfrak{p}\),
\begin{equation}\label{eq-2001103}
\begin{split}
\mathfrak{p}(\mathcal{E}):&=\mathbb{E}_{U}\left[\bbra{\rho_0}\mathcal{U}^\dag\mathcal{E}^\dag\mathcal{E}\mathcal{U}\kett{\rho_0}\right]\\
&=\begin{bmatrix}1/\sqrt{d}&\sqrt{1-1/d}\bm{\rho}'^{\dag}_0\end{bmatrix}
\begin{bmatrix}
\mathfrak{t}(\mathcal{E})^2+\mathcal{E}_{\textup{n}}^\dag\mathcal{E}_{\textup{n}} & *\,\mathbb{E}_{U}\left[\mathcal{U}_{\textup{u}}\right] \\
\mathbb{E}_{U}\left[\mathcal{U}_{\textup{u}}^\dag\right]*\,\, & \mathbb{E}_{U}\left[\mathcal{U}_{\textup{u}}^\dag\left(\mathcal{E}_{\textup{sdl}}^\dag\mathcal{E}_{\textup{sdl}}+\mathcal{E}_{\textup{u}}^\dag\mathcal{E}_{\textup{u}}\right)\mathcal{U}_{\textup{u}}\right]
\end{bmatrix}\begin{bmatrix}1/\sqrt{d}\\\sqrt{1-1/d}\bm{\rho}'_0\end{bmatrix}\\
&=\begin{bmatrix}1/\sqrt{d}&\sqrt{1-1/d}\bm{\rho}'^{\dag}_0\end{bmatrix}\begin{bmatrix}
\mathfrak{t}(\mathcal{E})^2+\mathcal{E}_{\textup{n}}^\dag\mathcal{E}_{\textup{n}} & 0\\
0 &  \frac{\tr\left(\mathcal{E}_{\textup{sdl}}^\dag\mathcal{E}_{\textup{sdl}}+\mathcal{E}_{\textup{u}}^\dag\mathcal{E}_{\textup{u}}\right)}{d^2-1} I
\end{bmatrix}\begin{bmatrix}1/\sqrt{d}\\\sqrt{1-1/d}\bm{\rho}'_0\end{bmatrix}\\
&=\frac{1}{d} \mathfrak{t}(\mathcal{E})^2 +\frac{1}{d}\tr\left(\mathcal{E}_{\textup{n}}^\dag\mathcal{E}_{\textup{n}}\right)+ \frac{1}{d(d+1)}\tr\left(\mathcal{E}_{\textup{sdl}}^\dag\mathcal{E}_{\textup{sdl}}+\mathcal{E}_{\textup{u}}^\dag\mathcal{E}_{\textup{u}}\right).
\end{split}
\end{equation}
Similarly, recall the definition of the orthogonality-preservation index \(\mathfrak{o}\),
\begin{equation}\label{eq-2001104}
\begin{split}
\mathfrak{o}(\mathcal{E}):&=\mathbb{E}_{U}\left[\bbra{\rho_1}\mathcal{U}^\dag\mathcal{E}^\dag\mathcal{E}\mathcal{U}\kett{\rho_0}\right]\\
&=\begin{bmatrix}1/\sqrt{d}&\sqrt{1-1/d}\bm{\rho}'^{\dag}_1\end{bmatrix}\begin{bmatrix}
\mathfrak{t}(\mathcal{E})^2+\mathcal{E}_{\textup{n}}^\dag\mathcal{E}_{\textup{n}} & 0\\
0 &  \frac{\tr\left(\mathcal{E}_{\textup{sdl}}^\dag\mathcal{E}_{\textup{sdl}}+\mathcal{E}_{\textup{u}}^\dag\mathcal{E}_{\textup{u}}\right)}{d^2-1} I
\end{bmatrix}\begin{bmatrix}1/\sqrt{d}\\\sqrt{1-1/d}\bm{\rho}'_0\end{bmatrix}\\
&=\frac{1}{d} \mathfrak{t}(\mathcal{E})^2 +\frac{1}{d}\tr\left(\mathcal{E}_{\textup{n}}^\dag\mathcal{E}_{\textup{n}}\right)- \frac{1}{(d-1)d(d+1)}\tr\left(\mathcal{E}_{\textup{sdl}}^\dag\mathcal{E}_{\textup{sdl}}+\mathcal{E}_{\textup{u}}^\dag\mathcal{E}_{\textup{u}}\right),
\end{split}
\end{equation}
where the last equality is because that \(0=\tr\left(\rho_1^\dag \rho_0\right)=1/d+(1-1/d)\bm{\rho}'^\dag_1\bm{\rho}'_0\). 

Then, by Eq. \eqref{eq-2001102}, Eq. \eqref{eq-2001103} and Eq. \eqref{eq-2001104}, we have
\begin{equation}
\mathfrak{u}'(\mathcal{E})=\frac{1}{d^2-1}\tr\left(\mathcal{E}_{\textup{u}}^\dag\mathcal{E}_{\textup{u}}\right)=\mathfrak{p}(\mathcal{E})-\mathfrak{o}(\mathcal{E})-\frac{d}{d-1}\mathfrak{s}(\mathcal{E})+\frac{1}{d-1}\mathfrak{t}(\mathcal{E})^2.
\end{equation}

Note that \(\mathfrak{s}(\mathcal{E})\) can be estimated to precision \(O(\epsilon)\) by a similar strategy to that for \(\mathfrak{p}(\mathcal{E})\), except that a Haar random unitary is applied to one of the states before the SWAP test or distributed quantum inner product estimation. This implies the same upper bounds also apply to \(\mathfrak{u}'(\mathcal{E})\).

\subsection{Lower bound}\label{sec-12162322}
Note that the depolarizing vs random unitary problem is reducible to the estimation problem for \(\mathfrak{u}'\) with constant precision. This is because:
\begin{equation}
\mathfrak{u}'(\mathcal{D})=0,\quad\quad \mathfrak{u}'(\mathcal{U})=1,
\end{equation}
where \(\mathcal{D}\) is the completely depolarizing channel and \(\mathcal{U}\) is a unitary channel. Therefore, the lower bound \(\Omega(\sqrt{d})\) also applies to the estimation problem of \(\mathfrak{u}'\) with incoherent access. On the other hand, recall from Problem \ref{pro-11232006}, \(\mathcal{E}_1,\mathcal{E}_2\) are both unital and trace-preserving channels, which implies that their nonunital block and state-dependent leakage block are zero, and the trace-preservation index \(\mathfrak{t}=1\). This means
\begin{equation}
\mathfrak{u}'(\mathcal{E}_2)-\mathfrak{u}'(\mathcal{E}_1)=\frac{d^2}{d^2-1}\left(\mathfrak{u}(\mathcal{E}_2)-\mathfrak{u}(\mathcal{E}_1)\right)=\Omega(\epsilon).
\end{equation}
Therefore, the lower bound \(\Omega(\epsilon^{-2})\) also applies to the estimation problem for \(\mathfrak{u}'(\mathcal{E})\) with coherent access. Combined with the previous result, the lower bound \(\Omega(\sqrt{d}+\epsilon^{-2})\) for incoherent access follows.

\end{document}